\documentclass[superscriptaddress]{revtex4}
\usepackage[dvips]{graphicx}
\usepackage{epsfig}
\usepackage{amsmath}
\usepackage{amssymb}
\usepackage{amsthm}
\usepackage{epsfig}
\usepackage{subfigure}
\usepackage{paralist}
\usepackage{calc}

\newtheorem{theorem}{Theorem}

\newtheorem{lemma}{Lemma}

\newtheorem{definition}{Definition}

\newcommand{\remove}[1]{}
\newcommand{\sweep}{\textbf{SWEEP} }
\newcommand{\chern}{\emph{Chernoff Bound} }
\newcommand{\surn}{\emph{surrounding neighborhood} }
\newcommand{\FourNeigh}{4\!-\!Neighbors }

\newcommand{\EightNeigh}{8\!-\!Neighbors }

\begin{document}

\title{The Cooperative Cleaners Problem in Stochastic Dynamic Environments}

\author{Eyal Regev [1], Yaniv Altshuler [2], Alfred M.Bruckstein [1]\\
{\ [1]} \\ Technion --- IIT \\
\texttt{\{regeve,freddy\}@cs.technion.ac.il}
\\
{\ [2]} \\\ MIT Media Lab \\
\texttt{yanival@media.mit.edu}} 
%

\begin{abstract}
In this paper we study the strengths and limitations of collaborative teams of simple agents. 
In particular, we discuss the efficient use of ``ant robots'' for covering a connected region 
on the $\mbox{\bf Z}^{2}$ grid, whose area is unknown in advance and which expands stochastically. 
Specifically, we discuss the problem where an initial connected region of $S_0$ boundary tiles 
expand outward with probability $p$ at every time step. On this grid region a group of $k$ 
limited and simple agents operate, in order to clean the unmapped and dynamically expanding 
region. A preliminary version of this problem was discussed in~\cite{dCC,new-IJRR}, involving 
a deterministic expansion of a region in the grid.In this work we extend the model and examine
cases where the spread of the region is done stochastically, where each tile has some probability 
$p$ to expand, at every time step. For this extended model we obtain an analytic probabilistic lower 
bounds for the minimal number of agents and minimal time required to enable a collaborative 
coverage of the expanding region, regardless of the algorithm used and the robots' hardware 
and software specifications. In addition, we present an impossibility result, for a variety 
of regions that would be impossible to completely clean, regardless of the algorithm used. 
Finally, we validate the analytic bounds using extensive empirical computer simulation results.\end{abstract}

\maketitle

\section{Introduction}

In this work we discuss in this work is the
\emph{Cooperative Cleaners} problem~---~a problem assuming a
regular grid of connected rooms (pixels), some of which are
`dirty' and the `dirty' pixels forming a connected region of the
grid. On this dirty grid region several agents operate, each
having the ability to `clean' the place (the `room', `tile',
`pixel' or `square') it is located in. We examine problems in which the agents work
in stochastic dynamic environments~---~where probabilistic changes in the environment 
may take place and are independent of, and certainly not caused by, the agents'
activity.
In the spirit of~\cite{Intro6} we consider simple robots with
only a bounded amount of memory
(i.e.~\emph{finite-state-machines}).

The static variant of this problem was introduced in \cite{CoopCleaner}, where a
cleaning protocol ensuring that a decentralized group of agents will jointly clean any given (and a-priori
unknown) dirty region. The protocol's performance, in terms of
cleaning time was fully analyzed and also demonstrated experimentally.

A dynamic generalization of the problem was later presented in \cite{new-IJRR}, in which a deterministic
expansion of the environment is assumed, simulating the spreading of a \emph{contamination} 
(or a spreading "danger" zone or \emph{fire}). Once again,
the goal of the agents is to clean the spreading contamination
in as efficiently as possible.

In this work we modify the 'dirty' region expansion model and add stochastic features to the spreadings of the region's cells. We formally define and analyze the \emph{Cooperative Cleaning} problem, this time under this stochastic generalization. 

We focus on a variant of the \emph{Cooperative Cleaning} problem, where the tiles have some probability to be contaminated by their  neighbors' contamination. This version of the problem has applications in the ``real'' world and in the computer network environments as well. For instance, one of the applications can be a distributed anti-virus software trying to overcome an epidemic malicious software attacking a network of computers. In this case, each infected computer has some probability to infect computers connected to it.


A more general paradigm of the cleaning problem is when the transformation of the contaminated area from one state to another is described in the form of some pre-defined function. For instance, following the previous example, we can say that the sub-network affected by the virus, is spreading by a certain rule. We can say that a computer will be infected by a virus with a certain probability, which depends on the number of the neighboring computers already infected. By defining rules for the contamination's spreading and cleaning, we can think of to this problem as a kind of Conway's ``\emph{Game of Life}'', where each cell in the game's grid spreads its ``seed'' to neighboring cells (or alternatively, ``dies'') according to some basic rules. 

While the problem posed in \cite{new-IJRR}, as well as the analysis methods used and the correctness proofs, were all deterministic, it is interesting to examine the stochastic variant of such algorithms. In this work we analyze and derive a lower bound on the expected cleaning time for $k$ agents running a cleaning protocol under a model, where every tile in the neighborhood of the affected region may become contaminated at every time step with some probability, the contamination coming from its dirty neighbors.

%

\section{Organization}
The main contributions of this paper are a bound on the contaminated region's size and a bound on the cleaning time which is presented in Section~\ref{l-bound}. We also present a method which bounds the cleaning time for a given desired probability. We then provide an impossibility result for the problem raised in Section~\ref{i-result}. The rest of paper is organized as followed: In Section~\ref{r-work} we survey some of the related works. In Section~\ref{b-def} we formalize the problem, giving the basic definitions needed for the later analysis, in Section~\ref{e-results} we present some of the experimental results and compare them to the analytic bounds and concluding in Section~\ref{conclu}.

\section{Related Work}
\label{r-work}
Significant research effort is invested in the design and simulation of
multi-agent robotics and intelligent swarm systems (see e.g.
\cite{Mastellone1,DeLoach1,Chalkiadakis1,GraphSearch1,CoopRobot1,AMAI_Intro}).

In general, most of the techniques used for the distributed coverage of some region are based on some sort of cellular decomposition. For example, in~\cite{primitive_static_cleaning} the area to be covered is divided between the agents based on their relative locations. In~\cite{butler1} a different decomposition method is being used, which is analytically shown to guarantee a complete coverage of the area. \cite{choset2} discusses two methods for cooperative coverage (one probabilistic and the other based on an exact cellular decomposition).

While some existing works concerning distributed (and decentralized) coverage present analytic proofs for the ability of the system to complete the task (for example, in~\cite{choset2,butler1,batalin1}), most of them lack analytic bounds for the coverage time (and often extensive amounts of empirical results on this are made available by extensive simulations). Although a proof for the coverage completion is an essential element in the design of a multi-agent system, analytic indicators for its efficiency are in our opinion of great importance. We provide such results, as bounds for the cleaning time of the agents, in Section~\ref{l-bound}.

An interesting work to mention in this context is that of Koenig and his collaborators~\cite{Svennebring1,similar_terrain_coverage}, where a swarm of ant-like robots is used for repeatedly covering an unknown area, using a real time search method called \emph{node counting}. By using this method, the robots are shown to be able to efficiently perform a coverage mission, and analytic bounds for the coverage time are discussed.

Another work discussing a decentralized coverage of terrains is presented in \cite{Zheng2}. This work examines domains with non-uniform traversability. Completion times are given for the proposed algorithm, which is a generalization of the forest search algorithm. In this work, though, the region to be searched is assumed to be known in advance~-~a crucial assumption for the search algorithm, which relies on a cell-decomposition procedure.

\emph{Vertex-Ant-Walk}, a variation of the node counting algorithm is presented in \cite{GraphSearch1} and is shown to achieve a coverage time of $O(n \delta_{G})$, where $\delta_{G}$ is the graph's diameter, which is based on a previous work in which a cover time of $O(n^{2} \delta_{G})$ was demonstrated \cite{Thrun-coverage}.
Another work called \emph{Exploration as Graph Construction}, provides a coverage of degree bounded graphs in $O(n^{2})$ time, is described in \cite{Exploration-Dudek}. Here a group of ant robots with a limited capability explores an unknown graph using special ``markers''.

Similar works concerning multi agents systems may be found
in~\cite{survailance2,similar2,similar1,Bejar1,primitive_static_cleaning,2010CSMAS,choset2,butler1,min1,batalin1}).

The \emph{Cooperative Cleaning} problem is also strongly related to the problem of distributed search after mobile and evading target(s) \cite{Koenig-movingtargets,Koenig-pursuits-complexity,UAV-ROBOTICA, ICINCO-UAV} or the problems discussed under the names of ``Cops and Robbers'' or ``Lions and Men'' pursuits \cite{Lions-Isaacs,lions-1992,Lions-Isler,Lions-Goldstein,Lions-Flynn,lion-klein}.

\section{Definitions}
\label{b-def}

In our work we will use some of the basic notations and definitions, that were used in~\cite{new-IJRR}, which we shell briefly review. As in the above mentioned previous works on this problem, we shall assume that the time is discrete.  

\begin{definition}
Let an undirected graph $G(V,E)$ describe the two dimensional integer grid $\mbox{\bf Z}^{2}$, whose vertices (or ``\emph{tiles}'') have a binary property called ``\emph{contamination}''. Let
$cont_{t}(v)$ denote the contamination state of the tile $v$ at time $t$, taking either the value
``\emph{on}'' (for ``dirty'' or ``contaminated'') or ``\emph{off}'' (for ``clean'').
\end{definition}

For two vertices $v,u \in V$, the edge $(v,u)$ may belong to $E$ at time $t$ only if both of the
following hold~:
\begin{inparaenum}[\itshape a\upshape)]
  \item $v$ and $u$ are $\FourNeigh$ in $G$.
  \item $cont_{t}(v) = cont_{t}(u) = on$.
\end{inparaenum}
This however is a necessary but not a sufficient condition as we elaborate below.

The edges of $E$ represent the connectivity of the contaminated region. At $t=0$ all the
contaminated tiles are connected, namely~:
\[ (v,u) \in E_{0} \iff
(v,u \emph{ are $\FourNeigh$ in $G$}) \wedge (cont_{0}(v) =
cont_{0}(u) = on)\]

Edges may be added to $E$ only as a result of a contamination
spread and can be removed only while contaminated tiles are
cleaned by the agents.

\begin{definition}
Let $F_{t}(V_{F_{t}},E_{t})$ be the contaminated sub-graph of $G$ at time $t$, i.e.~:
\[ V_{F_{t}} = \left\{v \in G \ | \ cont_{t}(v) = on \right\} \]
\end{definition}

We assume that $F_{0}$ is a single simply-connected component (the actions of the agents will be so
designed that this property will be preserved).

\begin{definition}
Let $\partial F$ denote the boundary of $F$. A tile is on the boundary if and
only if at least one of its $\EightNeigh$ is not in $F$, meaning~:
\begin{displaymath}
\partial F = \{v \ | \ v \in F \ \wedge \ \EightNeigh(v) \ \cap \
(G \ \setminus \ F) \ \neq \ \emptyset\}
\end{displaymath}
\end{definition}

\begin{definition} \label{def.st} Let $S_{t}$ denote the size of the dirty region $F$ at time $t$, namely the number of grid
points (or tiles) in $F_{t}$.
\end{definition}
Let a group of $k$ agents that can move on the grid $G$ (moving from a tile to its neighbor in one
time step) be placed at time $t_{0}$ on $F_{0}$, at some point $p_{0} \in V_{F_{t}}$.

\begin{definition}
	Let us denote by \( \Delta{F_t} \) the \emph{potential boundary}, which is the maximal number of tiles which might be added to $F_t$ by spreading all the tiles of $\partial{F_t}$.
	\[ \Delta{F_t} \equiv \left\{v:\exists{u}\in\partial{F_t}\mbox{ and }v \in \FourNeigh(u)\mbox{ and }v\notin{F_t}\right\} \]
\end{definition}

As we are interested in the stochastic generalization of the dynamic cooperative cleaners model, we will assume that each tile in \( \Delta{F_t} \) might be contaminated with some probability $p$. In the model we will analyze later, we assume that the status variables of the tiles of \( \Delta{F_t} \) are independent from one another, and between time steps. 


\begin{definition}
	Let us denote by $\Phi_n\left(v\right)$ the \surn of a tile $v$, as the set of all the reachable tiles $u$ from $v$ within $n$ steps on the grid (namely, the ``digital sphere'' or radius $n$ around $v$). In this work we assume 4-connectivity among the region cells~---~namely, two tiles are considered as neighbors within one step iff the \emph{Manhattan distance} between them is exactly 1.
\end{definition}

The spreading policy, $\Xi\left(v\mbox{, }\phi\mbox{, }t\right)$, controls the contamination status of $v$ at time $t+1$, as a function of the contamination status of its neighbors in its $n$-th digital sphere, at time $t$. Notice the $\Xi\left(v\mbox{, }\phi\mbox{, }t\right)$ can be also non deterministic.

\begin{definition}
	Let us denote by $\Xi\left(v\mbox{, }\phi\mbox{, }t\right)$ the \emph{spreading policy} of $v$ as follows:
	$$ \Xi\left(v\mbox{, }\phi\mbox{, }t\right):\left(V,\ \left\{ \mbox{On, Off} \right\}^{\alpha}, \ \mathbb{N}\right)\rightarrow \left\{ \mbox{On, Off} \right\} $$
	Where $\alpha \equiv \left| \Phi_n\left(v\right) \right|$, $V$ denotes the vertices of the grid, $\left\{ \mbox{On, Off} \right\}^{\alpha}$ denotes the contamination status of the members of $\Phi_n\left(v\right)$, at time $t$ (for $t \in N$).
\end{definition}

A basic example of using the previous definition of $\Xi\left(\right)$ is the case of the \emph{deterministic model}, where at every $d$ time-steps the contamination spreads from all tiles in $\partial{F_t}$ to all the tiles in $\delta{F_t}$. This model can be defined using the $\Xi$ function, as follows:

For every tile $v$ we first define $\Phi_n\left(v\right)$ where $n$ equals to 1 and assuming 4-connectivity.
Then $\Xi\left(v\mbox{, }\phi\mbox{, }t\right)$ , for any time-step $t$ will be defined as follows:
$$\Xi\left(v\mbox{, }\phi\mbox{, }t\right) = \left\{
\begin{array}{l l}
\mbox{On} & \quad \mbox{if } t\mod d = 0 \mbox{ and } v \in \Delta{F_t}\\
\mbox{Off} & \quad \mbox{Otherwise}\\
\end{array} \right. $$
Notice that due to the fact that we are assuming that $v$ is in $\Delta{F_t}$, its \surn contains at least one tile with contamination status of \emph{On}.

An interesting particular case of the general $\Xi()$ function is the \emph{simple uniform probabilistic spread}. In this scenario, a tile in $V \in \Delta{F_t}$ becomes contaminated with some predefined probability $p$, if and only if at least one of its $n$-th neighbors are contaminated at time step $t$. Using the $\Xi()$ function and the probability $p$, this can be formalized as follows~:

For every tile $v$ we first define $\Phi_n\left(v\right)$ where $n$ equals to 1 and assuming 4-connectivity.
$$\Xi\left(v\mbox{, }\phi\mbox{, }t\right) = \left\{
\begin{array}{l l}
\mbox{On} & \quad \mbox{with probability of } p \mbox{ if } v \in \Delta{F_t}\\
\mbox{Off} & \quad \mbox{Otherwise}\\
\end{array} \right. $$

This model can naturally also be defined simply as :
\[ \forall t \in \mathbb{N} \mbox{, } \forall v \in \Delta{F_t} \mbox{, }	Prob\left(cont_{t+1}(v) = On\right) = p \]

In our work we will focus on this model, while deriving the analytic bounds for the cleaning time.

\section{Lower Bound}
\label{l-bound}

\subsection{Direct Bound}
\label{d-bound}

In this section we form a lower bound on the cleaning time of any cleaning protocol preformed by $k$ agents. We start by setting a bound on the contaminated region's size at each time step, $S_t$. As we are interested in minimizing the cleaning time we should also minimize the contaminated region's area. Therefore we are interested in the minimal size of it, which achieved when the region's shape is sphere \cite{new-IJRR}. In our model each tile in the \emph{potential boundary}, $\Delta{F_t}$, has the same probability $p$ to be contaminated in the next time step. The whole probabilistic process in each time step is \emph{Binomial Distributed}, under the assumption that the spreading of each tile at any time step is independent from the spreadings of other tiles or from the spreadings of the same tile at different time steps. 

As we are interested in the lower bound of the contaminated region's size we will assume that the expected number of newly added tiles to the contaminated region is minimal, which occurs when the region's shape forms a digital sphere (as presented in \cite{DGCI}). Then we can compute the expectation of this process for a specific time step $t$. Therefore, the size of the \emph{potential boundary} is $\Delta{F_t} = 2\sqrt{2 \cdot S_t - 1}$ as shown in \cite{DGCI,Daniel}.

\begin{definition}
\label{def-xt}
Let us denote by $X_t$ the random variable of the actual number of added tiles to the contaminated region at time step $t$.
\end{definition}

Assuming the independence of tiles' contamination spreadings and given $S_t$, $X_t$ is  \emph{Binomial Distributed}, $X_t|S_t \sim B\left( \Delta{F_t},p\right)$, where each tile in the \emph{potential boundary} has the same probability $p$ to be contaminated. Therefore, we can say that the expectation of $X_t$ given $S_t$ is $\mu = E\left(X_t | S_t \right) = p \cdot \Delta{F_t}$.

Notice that occasionally the number of new tiles added to the contaminated region may be below $\mu$. As we are interested in a lower bound, we should take some $\mu ' < \mu$  such that: $Pr[X_t < \mu ' | S_t] < \epsilon$, meaning that the probability that the number of the newly added tiles to the contaminated region is smaller than $\mu '$ is extremely small (tends to 0).

In order to bound $X_t$ by some $\mu '$ we shall use the \chern, where $\delta$ is the desired distance from the expectation, as follows~:
$$ Pr\left[X_t < (1 - \delta) \mu | S_t \right] < e^{-\frac{\delta^2 \mu}{2}}$$

\begin{definition}
\label{def-qt}
Let us denote by $q_t$ the probability that at time step $t$, the size of the added tiles to the contaminated region is not lower than $\mu'=(1-\delta)\mu$ and it can be written as follows~:
$$ q_t = (1- Pr\left[X_t < (1 - \delta) \mu | S_t\right]) $$
\end{definition}

\begin{theorem}
Using any cleaning protocol, the area of the contaminated region at time step $t$ can be recursively
lower bounded, as follows~:
\label{theorem.DynamicProbLower.LowerBound} %
\begin{displaymath}
Pr \left[S_{t+1} \geq S_{t} - k + \left\lfloor 2 \cdot \left( 1-\delta \right) p \cdot \sqrt{2 \cdot (S_{t} - k) - 1} \right\rfloor | S_t\right] \geq q_t
\end{displaymath}
\end{theorem}

\begin{proof} Notice that a lower bound for the contaminated region's size can be obtained by assuming that the agents are
working with maximal efficiency, meaning that each time step every agent cleans exactly one tile. 

In each step, the agents clean another portion of $k$ tiles, but the remaining contaminated tiles spread their contamination to their $\FourNeigh$ and cause new tiles to be contaminated. 

Lets us denote by the random variable $S_{t+1}$ the number of contaminated tiles in the next time step. Using Definition~\ref{def-xt} we can express $S_{t+1}$ as follows~:
$$ S_{t+1} = S_t - k + X_t$$
Lets first bound the number of the added tiles using the \chern. As $X_t$ given $S_t$ is  \emph{Binomial Distributed}, $X_t | S_t\sim B\left( \Delta{F_t},p\right)$ and $\mu = E\left(X_t| S_t\right) = p \cdot \Delta{F_t}$. Using \chern we know that:
$$ Pr\left[X_t < (1 - \delta) \mu | S_t\right] < e^{-\frac{\delta^2 \mu}{2}} \Rightarrow Pr\left[X_t < (1 - \delta) p \cdot \Delta{F_t}| S_t\right] < e^{-\frac{\delta^2 \cdot p \cdot \Delta{F_t}}{2}} $$
Assigning $X_t = S_{t+1}- S_t +k$ from former definition of $S_{t+1}$, we get:
$$ Pr\left[S_{t+1}- S_t + k < (1 - \delta) p \cdot \Delta{F_t} | S_t\right] < e^{-\frac{\delta^2 \cdot p \cdot \Delta{F_t}}{2}} $$
As we are interested in the \emph{minimal} number of tiles which can become contaminated at this stage. The minimal number of $\FourNeigh$ of any number of tiles is achieved when the tiles are organized in the shape of a ``digital sphere'' (see~\cite{DGCI,Daniel})~-~i.e. the \emph{potential boundary} is $\Delta{F_t} = 2\sqrt{2 \cdot (S_t-k) - 1}$. Assigning $\Delta{F_t}$ value:
$$ Pr\left[S_{t+1} < S_t - k + (1 - \delta) p \cdot 2\sqrt{2 \cdot (S_t-k) - 1} | S_t\right] < e^{-\frac{\delta^2 \cdot p \cdot 2\sqrt{2 \cdot (S_t-k) - 1}}{2}} $$ 
As we are interested in the complementary event and using Definition~\ref{def-qt}
\begin{equation} 
\label{ineq}
Pr\left[S_{t+1} \geq S_t - k + (1 - \delta) p \cdot 2\sqrt{2 \cdot (S_t-k) - 1} | S_t\right] \geq 1- e^{-\frac{\delta^2 \cdot p \cdot 2\sqrt{2 \cdot (S_t-k) - 1}}{2}} = q_t
\end{equation}

As the number of tiles must be an integer value, we use $\left\lfloor (1 - \delta ) \cdot p \cdot 2\sqrt{2 \cdot S_{t} - 1} \right\rfloor$ to be on the safe side. Using inequality~\ref{ineq} we get~:

$$Pr \left[S_{t+1} \geq S_{t} - k + \left\lfloor 2 \cdot \left( 1-\delta \right) p \cdot \sqrt{2 \cdot (S_{t} - k) - 1} \right\rfloor | S_t\right] \geq q_t$$

\end{proof} 

\begin{figure}
\centering
\subfigure[A lower bound for the contaminated region $S_t$, the area at time $t$, for various values of $\delta$] 
{
    \label{subfig.stdeltas}
    \includegraphics[height=7.8cm,width=9cm]{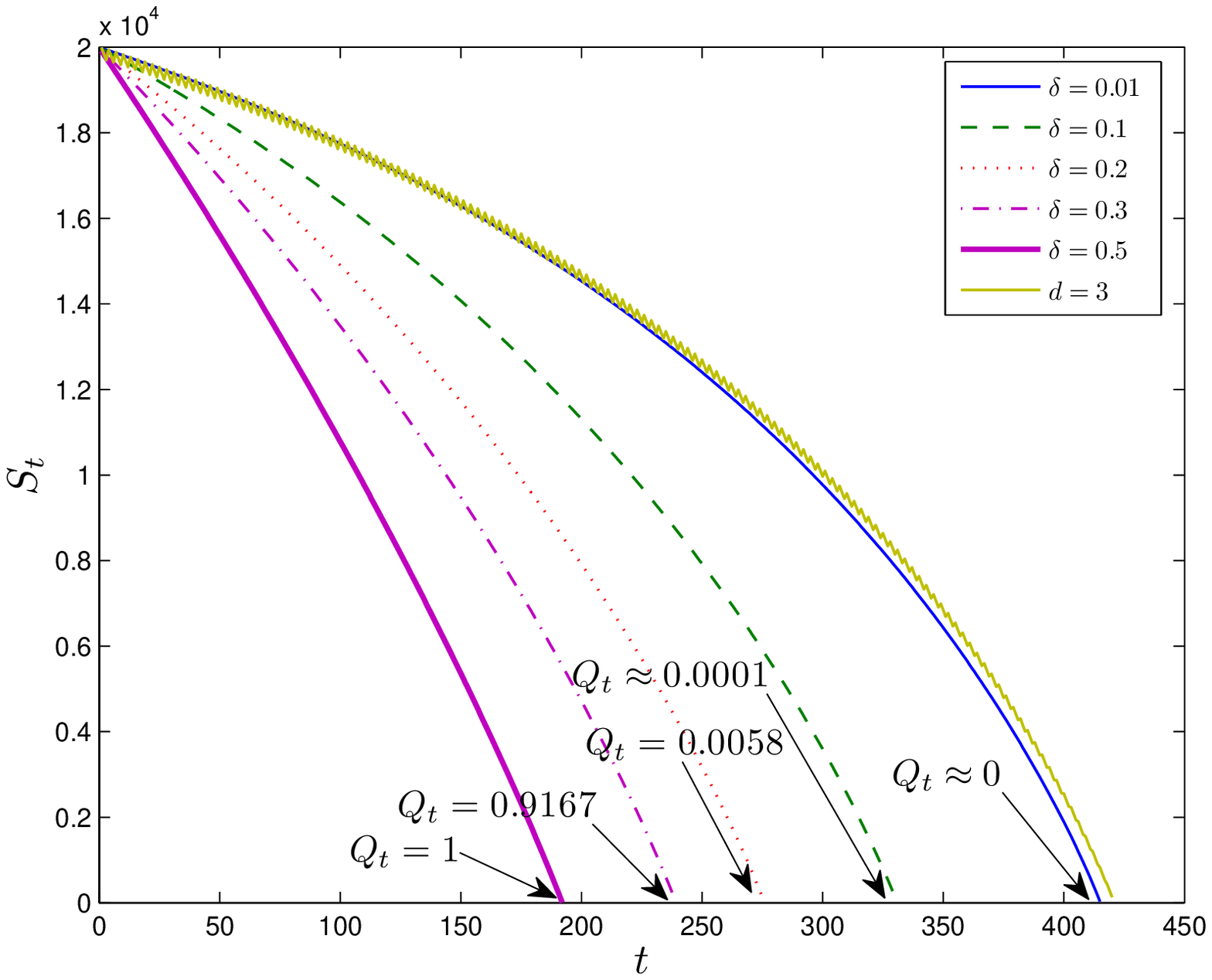}
}
\vspace{0.1cm}
\subfigure[A lower bound for the contaminated region $S_t$, the area at time $t$, for various number of agents $k$] 
{
    \label{subfig.stagents}
    \includegraphics[height=7.8cm,width=9cm]{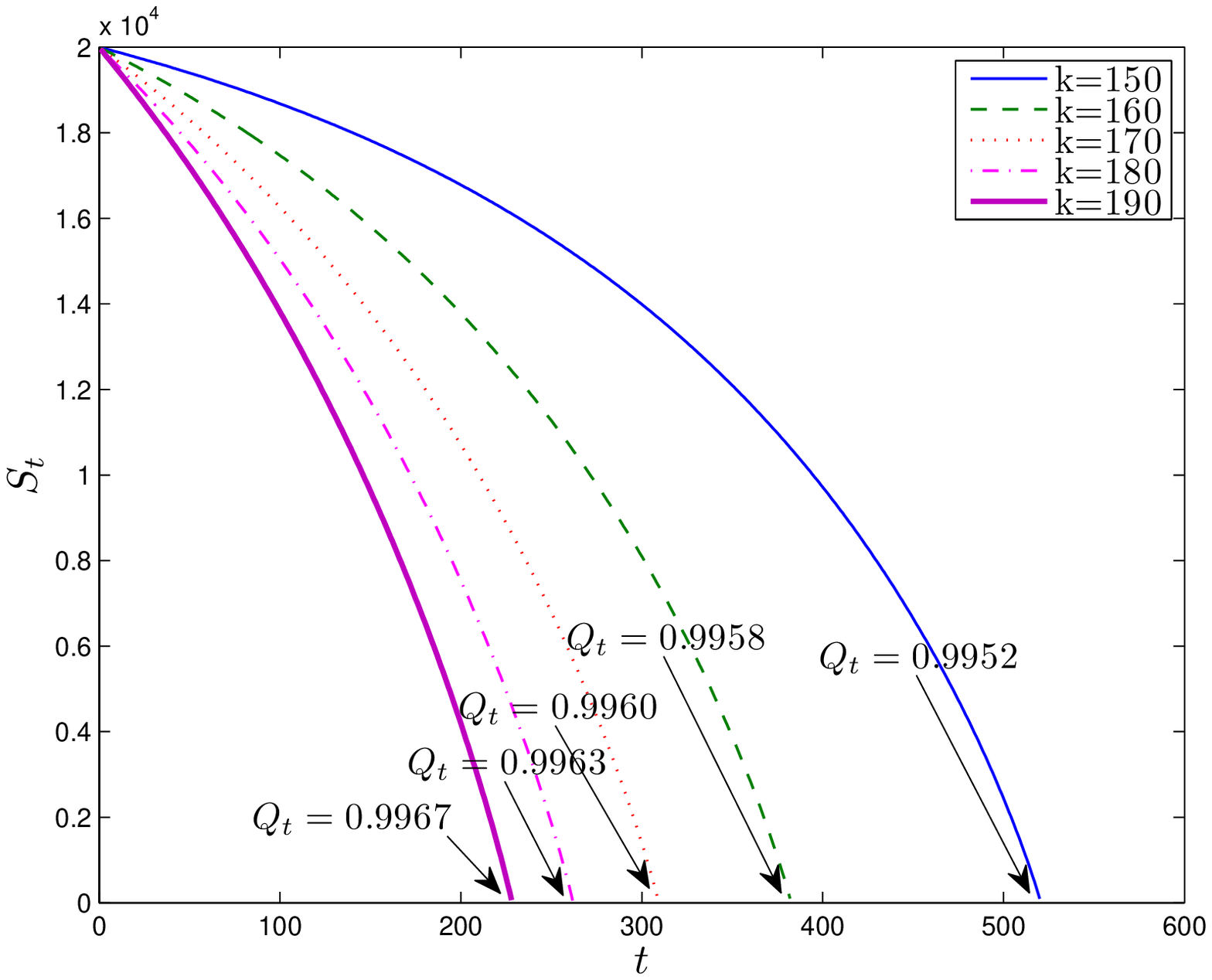}
}
\caption{An illustration of the bound presented in Theorem~\ref{theorem.DynamicProbLower.LowerBound}. In \subref{subfig.stdeltas} we can see the deterministic model (the Zig-Zag line) with spread each 3 time steps $(d=3)$ compare to the stochastic model with $p=1/3$ and $\delta \in \left[ 0.01, 0.1,0.2,0.3,0.5 \right]$ where both models have $k=150$ and start with $S_0=20000$. In \subref{subfig.stagents} we have the lower bound for $S_0 = 20000 \mbox{, } p = 0.5 \mbox{ and } \delta = 0.3$ for different number of agents $k \in \left[150,160,170,180,190\right]$.  }
\label{fig.DynamicProbLower.st}
\end{figure}

Notice that as illustrated in Figure~\ref{subfig.stdeltas}, which demonstrates the bound presented in Theorem~\ref{theorem.DynamicProbLower.LowerBound}, as $\delta$ decreases the produced bound for the stochastic model is closer to the bound of the deterministic model, for $d=\frac{1}{p}$.

\begin{definition}
\label{def.DynamicProbLower.qt}
	Let us denote by $Q_t$ the \emph{bound probability} that the contaminated region's size at time step $t$ will be [at least] $S_t$. $Q_t$ can be expressed as follows~:
	$$Q_t = \prod^{t}_{i=0} q_i$$
\end{definition}

Notice that for bounding the area of the region at time step $t$ using Theorem~\ref{theorem.DynamicProbLower.LowerBound} and Definition~\ref{def.DynamicProbLower.qt}, the bound, which will be achieved, will be in probability $Q_t$. We want to have $Q_t$ sufficiently high. 

We shall assume, for the sake of analysis, that the dynamic value of the area, $S_t$, is always kept not less than some $\hat{S} < S_{0} - k + \left\lfloor 2 \cdot \left( 1-\delta \right) p \cdot \sqrt{2 \cdot (S_{0} - k) - 1} \right\rfloor$ (as we want $S_1$ to be bigger or equal to $\hat{S}$). Then the next Lemma holds~:
\begin{lemma}
\label{lemma.DynamicProbLower.LowerBound}
For any $T \geq 1$, if for all $1 \leq t \leq T$ the contaminated region's size $S_t$ is always kept not less than some $\hat{S} < S_{0} - k + \left\lfloor 2 \cdot \left( 1-\delta \right) p \cdot \sqrt{2 \cdot (S_{0} - k) - 1} \right\rfloor$ then~:
\[
Q_T \geq \hat{Q}_T = \hat{q}^T = \left( 1 - e^{-\delta^2 \cdot p\cdot 2 \sqrt{2\cdot (\hat{S}-k) -1}} \right)^{T} 
\]
\end{lemma}

\begin{proof}
We will prove this Lemma by induction on T.
\begin{itemize}
\item \textit{Base step:}~For $T=1$, we can write the probabilities $Q_1$ and $\hat{Q}_1$ as follows~:
$$ 
Q_{1} = q_1 \geq 1 - e^{-\delta^2 \cdot p\cdot 2 \sqrt{2\cdot (S_1-k) -1}} 
$$
$$ 
\hat{Q}_1 = \hat{q} \geq 1 - e^{-\delta^2 \cdot p\cdot 2 \sqrt{2\cdot (\hat{S}-k) -1}} 
$$
As we assume that $S_t \geq \hat{S}$ than it is not hard to see that $Q_1 \geq \hat{Q}_1 $.
\item \textit{Induction hypothesis:}~We assume that for all $T \leq T'$ holds $Q_T \geq \hat{Q}_T$.
\item \textit{Induction step:}~We will prove that for  $T = T'+1$ holds $Q_T \geq \hat{Q}_T$. From definition~\ref{def.DynamicProbLower.qt} we can write the probability $Q_T$ as follows~:
$$Q_T =  Q_{T'+1} = Q_{T'}\cdot q_{T'+1}$$
By the induction hypothesis we know that for all  $T \leq T'$ holds $Q_T \geq \hat{Q}_T$ than we can rewrite $Q_{T'+1}$ and also writing $\hat{Q}_T$ for $T = T'+1$ we get that~:
$$
Q_T =  Q_{T'+1} \geq \hat{Q}_{T'}\cdot q_{T'+1}
$$
$$
\hat{Q}_T =  \hat{Q}_{T'+1} = \hat{Q}_{T'}\cdot \hat{q}
$$
As we want to compare these probabilities and to prove that $Q_T \geq \hat{Q}_T$ all we need to show is that $q_{T'+1} \geq \hat{q}$. As we assume that for all $1 \leq t \leq T'+1$ holds that $S_t \geq \hat{S}$ and particularly for $t=T'+1$, than it is not hard to see that $q_{T'+1} \geq \hat{q}$ and therefore $Q_T \geq \hat{Q}_T$.
\end{itemize}

\end{proof}

\begin{theorem} 
For any contaminated region of size $S_0$, using any cleaning protocol, the probability that $S_{\hat{\tau_{\delta}}}$, the contaminated area at time step $t = \hat{\tau_{\delta}}$, is greater or equal to some $\hat{S} < S_{0} - k + \left\lfloor 2 \cdot \left( 1-\delta \right) p \cdot \sqrt{2 \cdot (S_{0} - k) - 1} \right\rfloor$ can be lower bounded, as follows~:
\label{theorem.DynamicProbLower.LowerBound3} %
$$
Pr\left[S_{\hat{\tau_{\delta}}} \geq \hat{S} \right] \geq \left(1 - e^{-\delta^2 \cdot p\cdot 2 \sqrt{2\cdot (\hat{S}-k) -1}}\right)^{\hat{\tau_{\delta}}}
$$
where~:
$$  
\hat{\tau_{\delta}} \triangleq \frac{\sqrt{\varpi \cdot\left(\hat{S} - k - \frac{1}{2}\right)} - \sqrt{\varpi \cdot\left(S_0 - k - \frac{1}{2}\right)} + \ln \left(\frac{ \sqrt{\varpi\cdot\left(\hat{S} - k - \frac{1}{2}\right)} - \frac{k}{2}}{ \sqrt{\varpi \cdot\left(S_0 - k - \frac{1}{2}\right)} - \frac{k}{2}}\right)^{\frac{k}{2}}}{\varpi}
$$
and 
$$
\varpi \triangleq 2 (1-\delta)^2 \cdot p^2 
$$
\label{theorem.DynamicProbLower.LowerBound2} %
\end{theorem}

\begin{proof} 
Observe that by denoting $y_{t} \triangleq S_{t}$
Theorem~\ref{theorem.DynamicProbLower.LowerBound} can be written as~:
\begin{displaymath}
y_{t+1} - y_{t} \geq \left\lfloor 2\cdot (1-\delta) \cdot p\sqrt{2 \cdot (y_{t} - k) - 1}\right\rfloor - k
\end{displaymath}

Searching for the minimal area we can look at the equation~:
\begin{displaymath}
y_{t+1} - y_{t} = \left\lfloor 2\cdot (1-\delta) \cdot p \sqrt{2 \cdot (y_{t} - k) - 1}\right\rfloor - k
\end{displaymath}

By dividing both sides by $\Delta t = 1$ we obtain~:
\begin{equation}
\label{equation_lowerbound_y'} %
y_{t+1} - y_{t} \triangleq y' = \left\lfloor \sqrt{(1-\delta)^2 \cdot p^2 \cdot 8 \left[ y -  \left(k + \frac{1}{2}\right) \right]}\right\rfloor - k
\end{equation}

Notice that the values of $y'$, the derivative of the change in the region's size, might be positive (stating an increase in the area), negative
(stating a decrease in the area), or complex numbers (stating that the area is smaller than $k$, and will therefore be cleaned before the next time step).

Let us denote $x^{2} \triangleq (1-\delta)^2 \cdot p^2 \cdot 8 \left[ y -  \left(k + \frac{1}{2}\right) \right]$. After calculating the
derivative of both sides of this expression we see that~:
\begin{displaymath}
2x \cdot x' = (1-\delta)^2 \cdot p^2 \cdot 8 y'
\end{displaymath}
and after using the definition of $y'$ of
Equation~\ref{equation_lowerbound_y'} we see that~:
\begin{eqnarray}
 \label{equation_lowerbound_x'}
2x \cdot \frac{dx}{dt} = 2x \cdot x' &=& (1-\delta)^2 \cdot p^2 \cdot 8 \left(\left\lfloor \sqrt{(1-\delta)^2 \cdot p^2 \cdot 8 \left[ y -  \left(k + \frac{1}{2}\right) \right]}\right\rfloor - k \right) \nonumber \\ &\leq&  (1-\delta)^2 \cdot p^2 \cdot 8\left(x - k\right)
\end{eqnarray}

From Equation~\ref{equation_lowerbound_x'} a definition of $dt$
can be extracted~:
\begin{eqnarray*}
dt &\geq& \frac{1}{8 \cdot (1-\delta)^2 \cdot p^2} \cdot \frac{2x}{x - k} dx \\ &\geq& \frac{1}{4 (1-\delta)^2 \cdot p^2} \cdot \frac{x - k+ k}{x - k} dx \geq \frac{1}{4 (1-\delta)^2 \cdot p^2} \left(1 + \frac{k}{x - k} \right)dx
\end{eqnarray*}

The value of $x$ can be achieved by integrating the previous expression as follows (notice that we
are interested in the equality of the two expressions)~:
\begin{displaymath}
\int_{t_{0}}^t dt = \int_{x_{0}}^x  \frac{1}{4 (1-\delta)^2 \cdot p^2} \left(1 + \frac{k}{x - k} \right)
dx
\end{displaymath}

After the integration we can see that~:
\begin{displaymath}
i \bigg|_{t_{0}}^t = \frac{1}{4 (1-\delta)^2 \cdot p^2} \left(x + k \ln \left(x - k \right) \right)
\bigg|_{x_{0}}^x
\end{displaymath}
and after assigning $t_{0} = 0$~:
\begin{displaymath}
4 (1-\delta)^2 \cdot p^2 \cdot t = x - x_{0} + k \ln\frac{x - k}{x_{0} - k}
\end{displaymath}

Returning back to $y$ and using $\varpi$ definition we get~:
\[
\varpi \cdot t = \sqrt{\varpi\left(y - k - \frac{1}{2}\right)} - \sqrt{\varpi\left(y_0 - k - \frac{1}{2}\right)} + \ln \left(\frac{ \sqrt{\varpi\left(y - k - \frac{1}{2}\right)} - \frac{k}{2}}{ \sqrt{\varpi\left(y_0 - k - \frac{1}{2}\right)} - \frac{k}{2}}\right)^{\frac{k}{2}}
\]
Returning to the original size variable $S_{t}$, we see that~:
\begin{equation}
\label{eq.DynamicProbLower.LowerBound3}
\varpi \cdot t = \sqrt{\varpi\left(S_t - k - \frac{1}{2}\right)} - \sqrt{\varpi\left(S_0 - k - \frac{1}{2}\right)} + \ln \left(\frac{ \sqrt{\varpi\left(S_t - k - \frac{1}{2}\right)} - \frac{k}{2}}{ \sqrt{\varpi\left(S_0 - k - \frac{1}{2}\right)} - \frac{k}{2}}\right)^{\frac{k}{2}}
\end{equation}

Defining that $\hat{\tau_{\delta}}=t$ and combining EQ.~\ref{eq.DynamicProbLower.LowerBound3} with Lemma~\ref{lemma.DynamicProbLower.LowerBound} knowing that $S_{t'} \geq \hat{S}$ for all $ 1 \leq t' \leq \hat{\tau_{\delta}}$ we get the following inequality:

$$
Pr\left[S_{\hat{\tau_{\delta}}} \geq \hat{S} \right] \geq \left(1 - e^{-\delta^2 \cdot p\cdot 2 \sqrt{2\cdot (\hat{S}-k) -1}}\right)^{\hat{\tau_{\delta}}}
$$
where~:
$$  
\hat{\tau_{\delta}} \triangleq \frac{\sqrt{\varpi \cdot\left(\hat{S} - k - \frac{1}{2}\right)} - \sqrt{\varpi \cdot\left(S_0 - k - \frac{1}{2}\right)} + \ln \left(\frac{ \sqrt{\varpi\cdot\left(\hat{S} - k - \frac{1}{2}\right)} - \frac{k}{2}}{ \sqrt{\varpi \cdot\left(S_0 - k - \frac{1}{2}\right)} - \frac{k}{2}}\right)^{\frac{k}{2}}}{\varpi}
$$
and 
$$
\varpi \triangleq 2 (1-\delta)^2 \cdot p^2 
$$
\end{proof} 

In Theorem~\ref{theorem.DynamicProbLower.LowerBound2} we can guarantee with high probability of $\hat{Q}_{\tau_{\delta}}$ that the contamination region's size will not be lower than $\hat{S}$~-~ namely for any time step $t > \tau_{\delta}$ the probability $Pr[S_t \geq \hat{S}]$ is getting lower and therefore the probability that the agents will succeed in cleaning the contaminated area is increasing. We are showing that by choosing small enough $\hat{S}$ so we know that the agents will succeed in cleaning the rest of the 'dirty' region, we will be able to guarantee with high probability the whole cleaning of the 'dirty' region. For example, choosing $\hat{S}$ to be in $o(k)$ will assure that for $S_t \leq \hat{S} \leq c \cdot k$ for some small constant $c$, the rest of the contaminated region will be cleaned in at most $c$ time steps by the $k$ cleaning agents.

\begin{figure}
\centering
\subfigure[The bound on $t$ for various values of $\delta$ as a function of $\hat{S}$] 
{
    \label{subfig.tdeltas}
    \includegraphics[height=8cm,width=9cm]{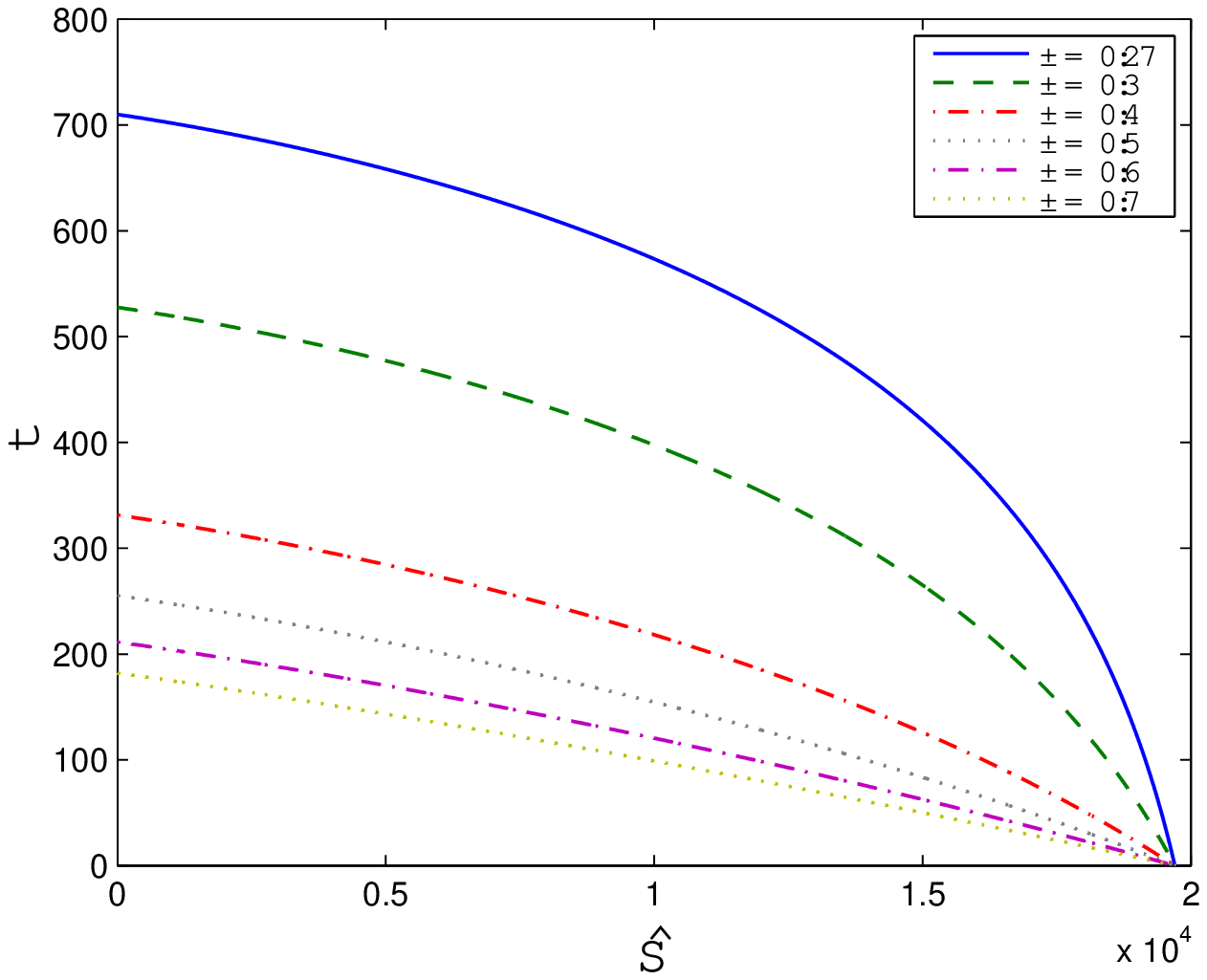}
}
\vspace{0.1cm}
\subfigure[The bound on $t$ for various number of agents $k$ as a function of $\hat{S}$] 
{
    \label{subfig.tagents}
    \includegraphics[height=8cm,width=9cm]{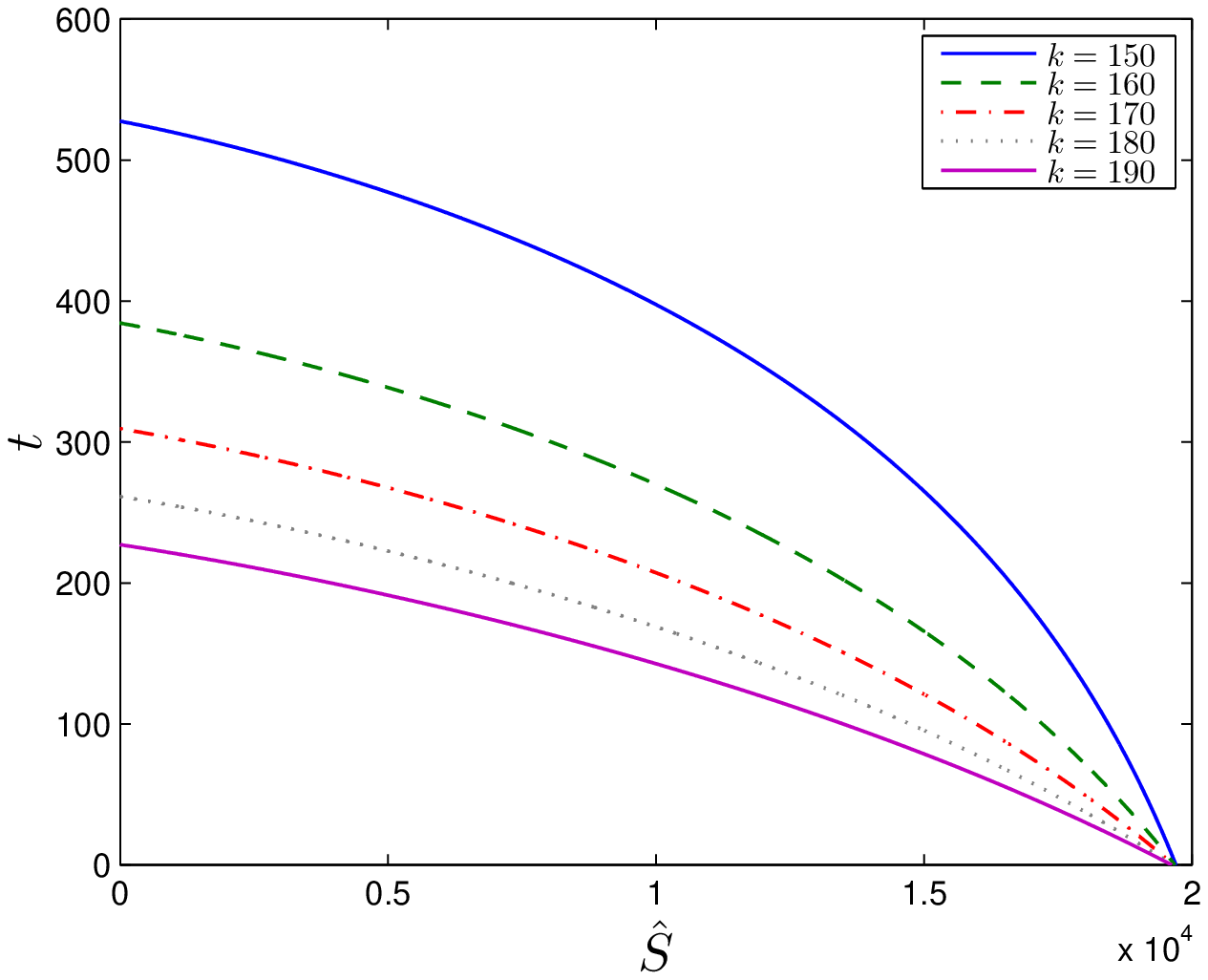}
}
\caption{An illustration of the bound presented in Theorem~\ref{theorem.DynamicProbLower.LowerBound2} of the cleaning time $t$ in order to reach $\hat{S}$. In \subref{subfig.tdeltas} we can see the bound on the cleaning time where $p=0.5$ and $\delta \in \left[ 0.27, 0.3, 0.4,0.5,0.6,0.7 \right]$ where the cleaning done by $k=150$ agents and starting with $S_0=20000$. In \subref{subfig.tagents} we have the lower bound on the cleaning time for $S_0 = 20000 \mbox{, } p = 0.5 \mbox{ and } \delta = 0.3$ for different number of agents $k \in \left[150,160,170,180,190\right]$.  }
\label{fig.DynamicProbLower.t}
\end{figure}

\begin{figure}
\centering
\subfigure[The bound probability $Q_t$ as a function of $\hat{S}$ for various $\delta$ values] 
{
    \label{subfig.t-delta}
    \includegraphics[height=8cm,width=9cm]{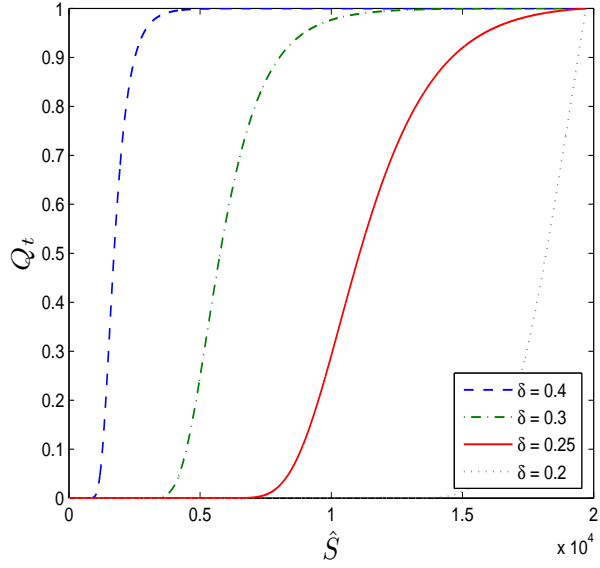}
}
\vspace{0.1cm}
\subfigure[The bound probability $Q_t$ as a function of $\hat{S}$ for various $p$ values] 
{
    \label{subfig.t-p}
    \includegraphics[height=8cm,width=9cm]{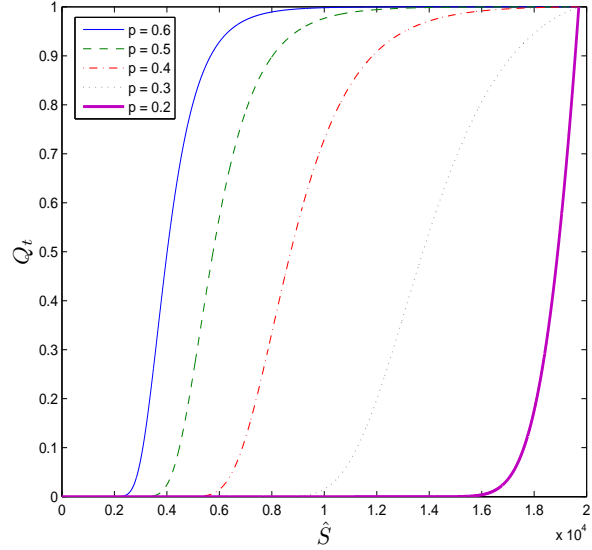}
}
\caption{An illustration of the probability produced by the bound presented in Theorem~\ref{theorem.DynamicProbLower.LowerBound2}. In \subref{subfig.t-delta} we can see the bound probability $Q_t$ for initial region's size $S_0 = 20000$, spreading probability $p= 0.5$ and number of agents $k=150$ for the following of values of $\delta \in \left[0.4,0.3,0.25,0.2\right]$. In \subref{subfig.t-p} we can see $Q_t$ for $S_0 = 20000$, $\delta = 0.3$ and $k=150$ for the following of values of $p \in \left[0.6,0.5,0.4,0.3\right]$}
\label{fig.DynamicProbLower.tp} 
\end{figure}

An illustration of the bound on the cleaning time, as presented in Theorem~\ref{theorem.DynamicProbLower.LowerBound2}, is shown in Figure~\ref{fig.DynamicProbLower.t} and the corresponding bound probability in Figure~\ref{fig.DynamicProbLower.tp}. Notice that as $\hat{S}$ and $\delta$ increase the cleaning time decreases.

\subsection{Using The Bound}
\label{u-bound}

In Theorem~\ref{theorem.DynamicProbLower.LowerBound2} we presented a bound which guarantees that the contaminated region's size will not be smaller than some predefined size $\hat{S}$ with the bound probability, $Q_t$. We suggest a method which make this bound useful when one willing to be guaranteed of successfully cleaning of the contaminated area with some desired probability $Q_t$ in a certain model, where the initial  contaminated region's size is $S_0$, each one of the tiles in the surrounding neighborhood of the contaminated area has a probability $p$ to be contaminated by each one of its neighbors's contamination spreads and with $k$ cleaning agents. 

Notice that the only free variables left in the bound are the analysis parameters $\delta$ and $\hat{S}$. We should also notice the fact that as $\delta$ decreases the ``usefulness'' of the bound decreases (see Figure~\ref{fig.DynamicProbLower.tp}) because when $\delta$ is closer to 0, the model tern to the deterministic variant of the \emph{cooperative cleaning} problem. In this variant of the problem the bound, as presented in Theorem~\ref{theorem.DynamicProbLower.LowerBound2}, will ``predict'' that the contamination will spread exactly by the \emph{potential boundary} mean in every step (as shown in~\ref{d-bound}). Furthermore, as $\delta$ increases, although the ``usefulness'' of the bound increases, the predicted bound is the naive one, where at each step there are no spreads and all the $k$ agents clean perfectly~-~i.e. the size of the contaminated region at time step $t$ will be exactly $S_t = S_0 - t\cdot k$ (which can be guaranteed in high probability).

We suggest the following method, in order to eliminate the need to identify the analysis parameters. Once someone willing to use this bound he should provide the desired bound probability~-~$Q_t$ and the parameter of the model. Then for each value of $\delta$ in the range of $\left[0,1\right]$ he should find the corresponding value of the minimal $\hat{S}$ which satisfies the inequality $Pr\left[S_t \geq \hat{S}\right] \geq Q_t$ as illustrated in Figure~\ref{fig.DynamicProbLower.Use.shat2delta}. As $\delta \in \mbox{\bf R}$~-~i.e. a real number, once using this method we should choose the granularity of $\delta$ for which calculate the appropriate $\hat{S}$.

\begin{figure}
	\centering
	\subfigure[$\hat{S}$ as a function of $\delta$ for various bound probability $Q_t$ values] 
{
    \label{subfig.DynamicProbLower.Use.shat2delta}
    \includegraphics[height=8cm,width=9cm]{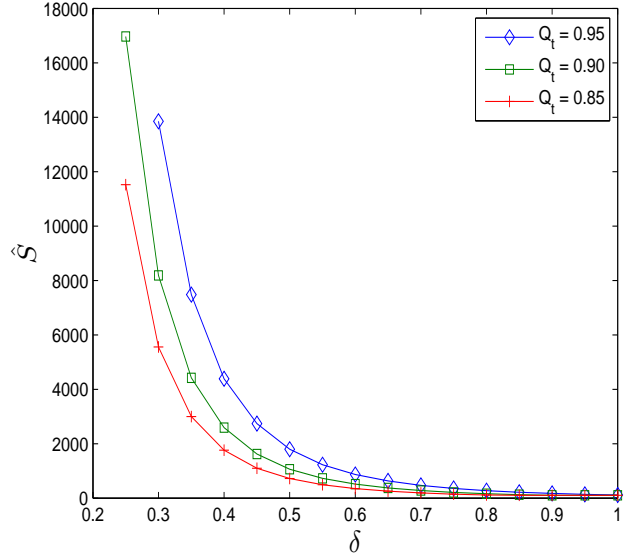}
}
\vspace{0.1cm}
\subfigure[The region's size $S_t$ for various values of $<Q_t,\hat{S}>$ pairs] 
{
    \label{subfig.DynamicProbLower.Use.allSt}
    \includegraphics[height=8cm,width=9cm]{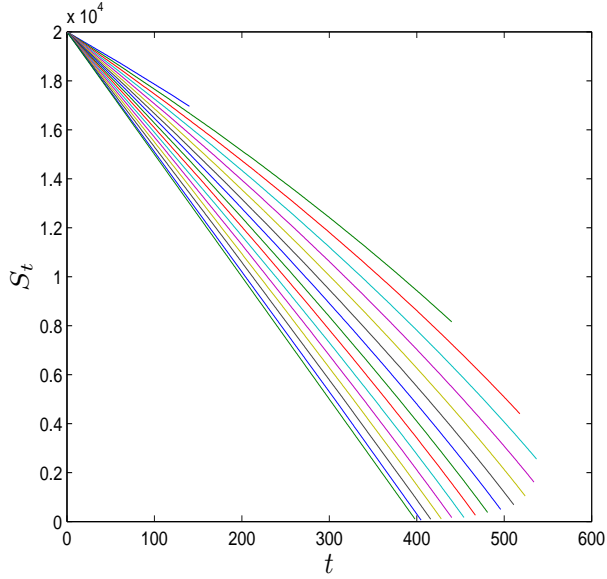}
}
	\caption{Figure~\subref{subfig.DynamicProbLower.Use.shat2delta} is an illustration of $\hat{S}$ as a function of $\delta$ using the bound presented in Theorem~\ref{theorem.DynamicProbLower.LowerBound2} for the following model parameters~--~$S_0=20000$, $p=0.1$ and $k=50$ for various values of $Q_t$. Figure~\subref{subfig.DynamicProbLower.Use.allSt} is an illustration of the cleaning process for various pairs of values of $\hat{S}$ and $\delta$ as shown in Figure~\subref{subfig.DynamicProbLower.Use.shat2delta} for the same model parameters with $Q_t=0.95$.}
	\label{fig.DynamicProbLower.Use.shat2delta} 
\end{figure}

Notice that there can be exists some minimal value of $\delta=\delta_{MIN}$, where for any value of $\delta < \delta_{MIN}$ there is no solution for the bound inequality. Also notice that there can be exists some maximal value of $\delta=\delta_{MAX}$, where for any value of $\delta > \delta_{MIN}$ the corresponding $\hat{S}$ is the same as for $\delta_{MAX}$.

Furthermore, for each pair of values of $\delta$ and $\hat{S}$ there exists its corresponding cleaning precess of $k$ agents with initial region's size $S_0$ as demonstrated in Figure~\ref{subfig.DynamicProbLower.Use.allSt}. Each curve bounds the cleaning process from $S_0$ to the applicable $\hat{S}$.

As we are interested in finding the tightest bound, looking at the frontier of the bounds, as shown in Figure~\ref{subfig.DynamicProbLower.Use.allSt}, we can combine the relevant curves to one comprehensive bound. This bounds integrates the bound for a specific range of $j$ values of $\delta \in \left[ \delta_{i_1},\delta_{i_2},...,\delta_{i_j}\right]$, where for each time step $t$ we choose the maximal $S_t$ as illustrated in Figure~\ref{subfig.DynamicProbLower.Use.combinedSt}.

Notice that the combined bound is independent of the selection of values for analysis parameters $\delta$ and $\hat{s}$. Also we can notice that this bound limits the contaminated region's size to some minimal $\hat{S}_{MIN}$
where we almost certain that the agents will succeed in terminating the cleaning process successfully. 

Notice that the inequality in Theorem~\ref{theorem.DynamicProbLower.LowerBound2} bounds the probability that the contaminated region's size at time step $t$ will not be smaller than $\hat{S}$ , therefore, an increase in the spreadings probability causes to an increase in the expected contaminated region's size and thus increases the probability $\left( \hat{q} \right)^t$ (as demonstrated in Figure~\ref{subfig.t-p}).

\begin{figure}
\centering

\subfigure[The combined bounds for various values of the bound probability $Q_t$] 
{
    \label{subfig.DynamicProbLower.Use.combinedSt}
    \includegraphics[height=8cm,width=9cm]{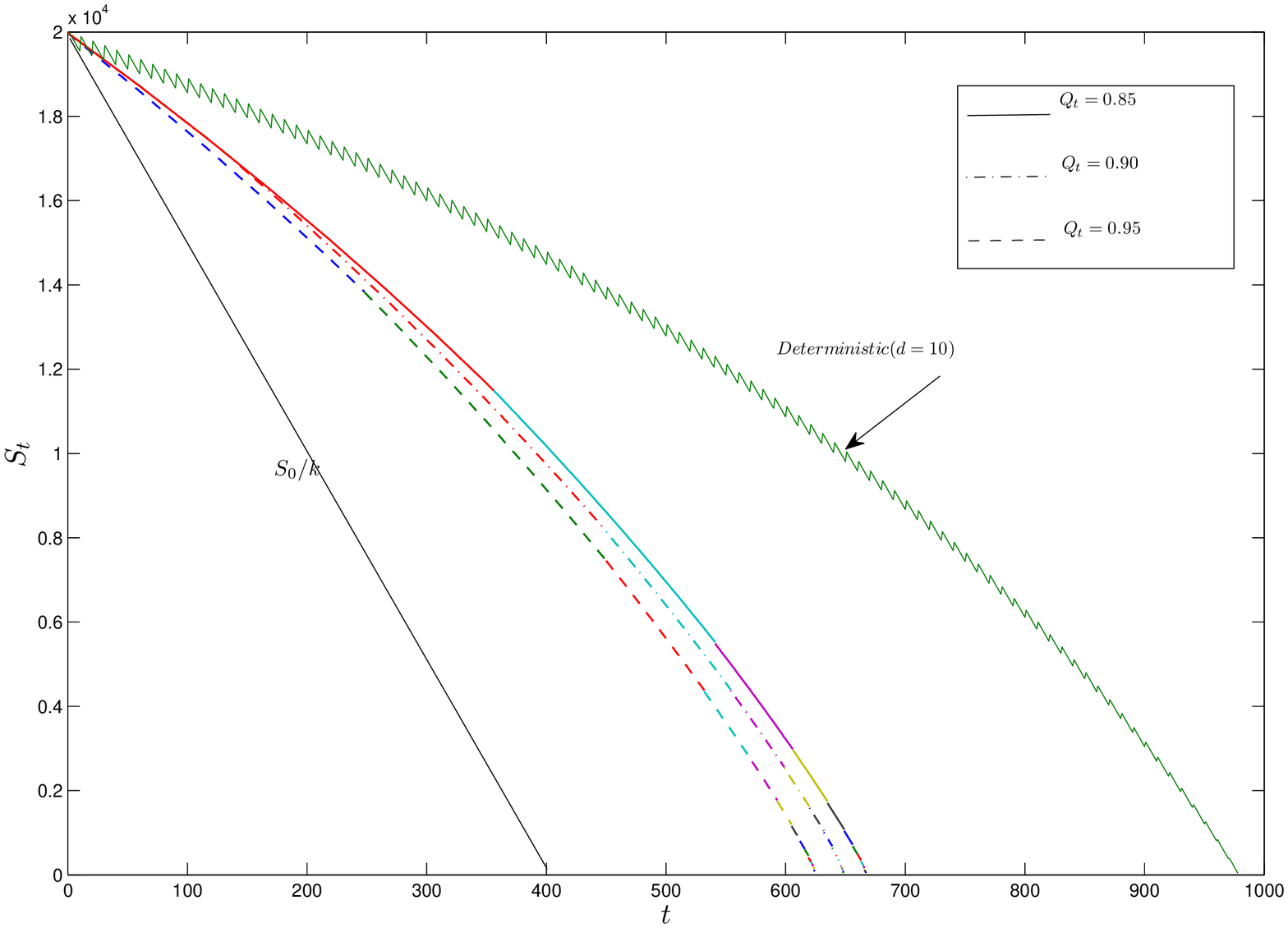}
}
\vspace{0.1cm}
\subfigure[The ratio between the bounds as function of the number of agents $k$] 
{
    \label{subfig.bounds.deltat2k}
    \includegraphics[height=8cm,width=9cm]{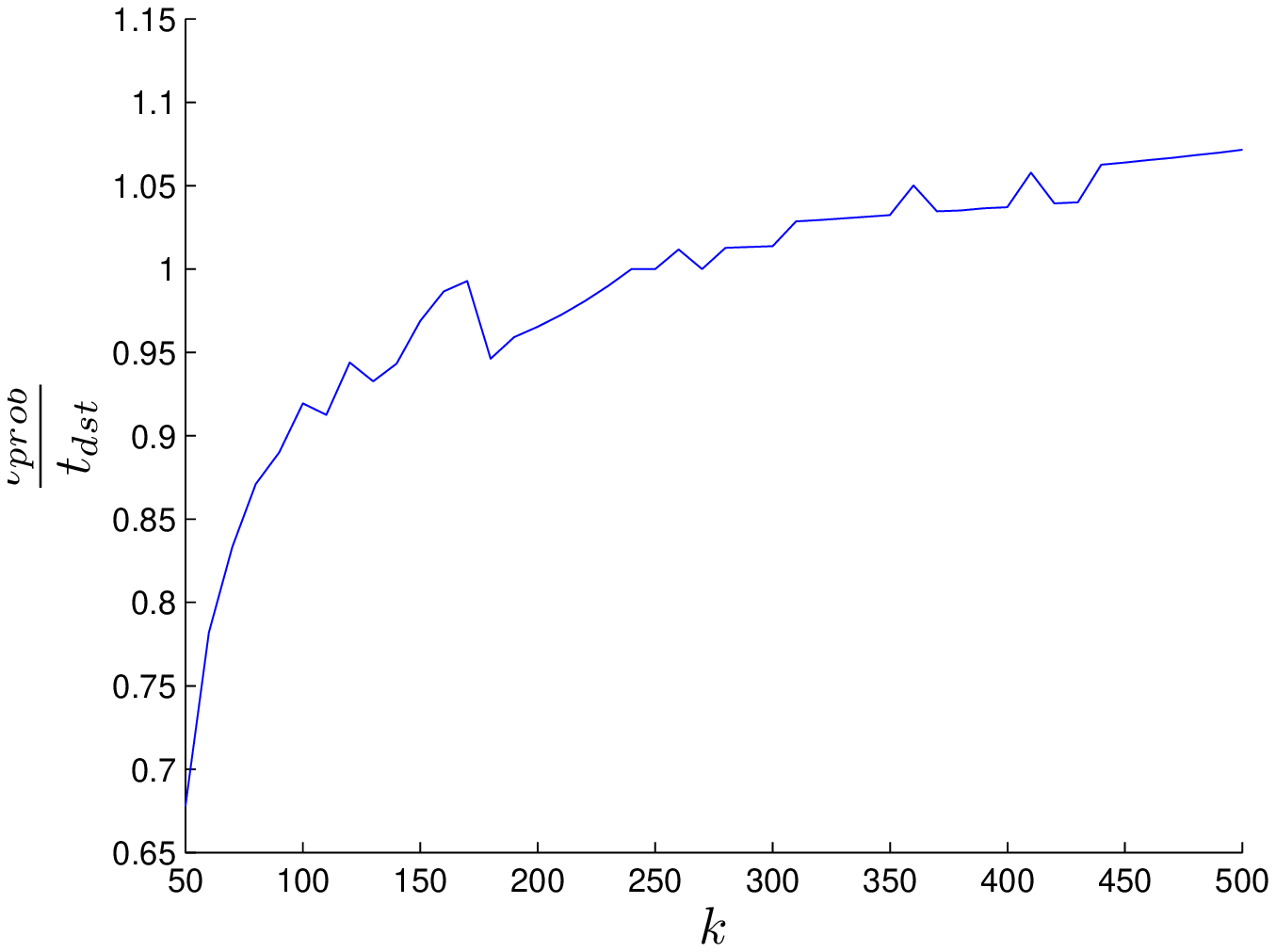}
}
\caption{Figure~\subref{subfig.DynamicProbLower.Use.combinedSt} An illustration of the combined bound for the bounds as shown in Figures~\ref{subfig.DynamicProbLower.Use.shat2delta} and \ref{subfig.DynamicProbLower.Use.allSt} for the same model parameters~--~$S_0=20000$, $p=0.1$ and $k=50$ for various values of $Q_t$ compared to the deterministic model with $d=10$ and to the naive bound $S_0 / k$.
Figure~\subref{subfig.bounds.deltat2k} compare the deterministic bound and the probabilistic bound as a function of the desired guaranteeing probability~-~$Q_t$ for various contaminated region's sizes and number of cleaning agents.}
\label{fig.DynamicProbLower.Use.boundsComparison} 
\end{figure}

\begin{figure}
\centering
\subfigure[The ratio between the bounds as function of the bound probability $Q_t$ for various number of agents $k$] 
{
    \label{subfig.bounds.deltat2qt.k}
    \includegraphics[height=8cm,width=9cm]{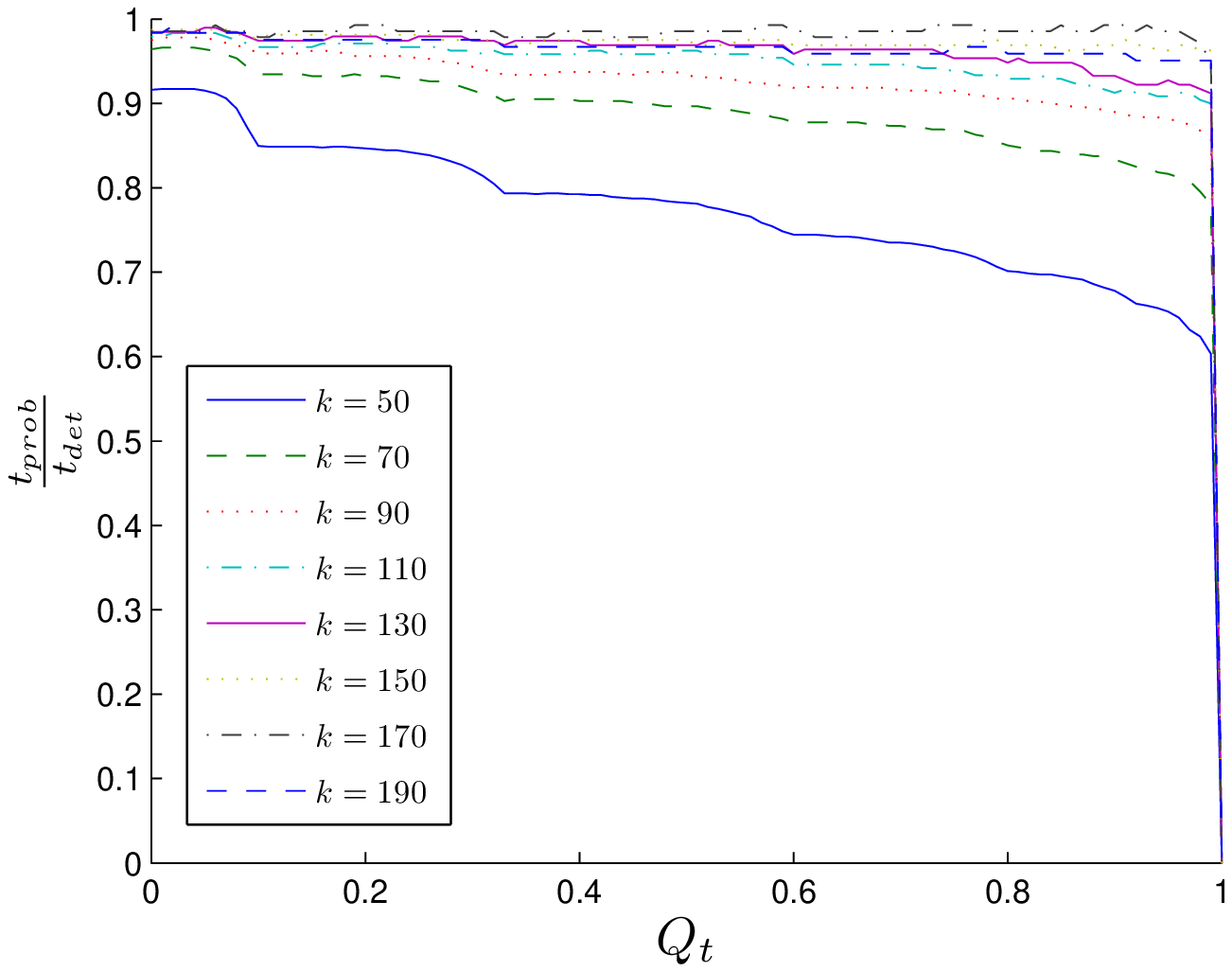}
}
\vspace{0.1cm}
\subfigure[The ratio between the bounds as function of the bound probability $Q_t$ for various values of initial region's size $S_0$] 
{
    \label{subfig.bounds.deltat2qt.s0}
    \includegraphics[height=8cm,width=9cm]{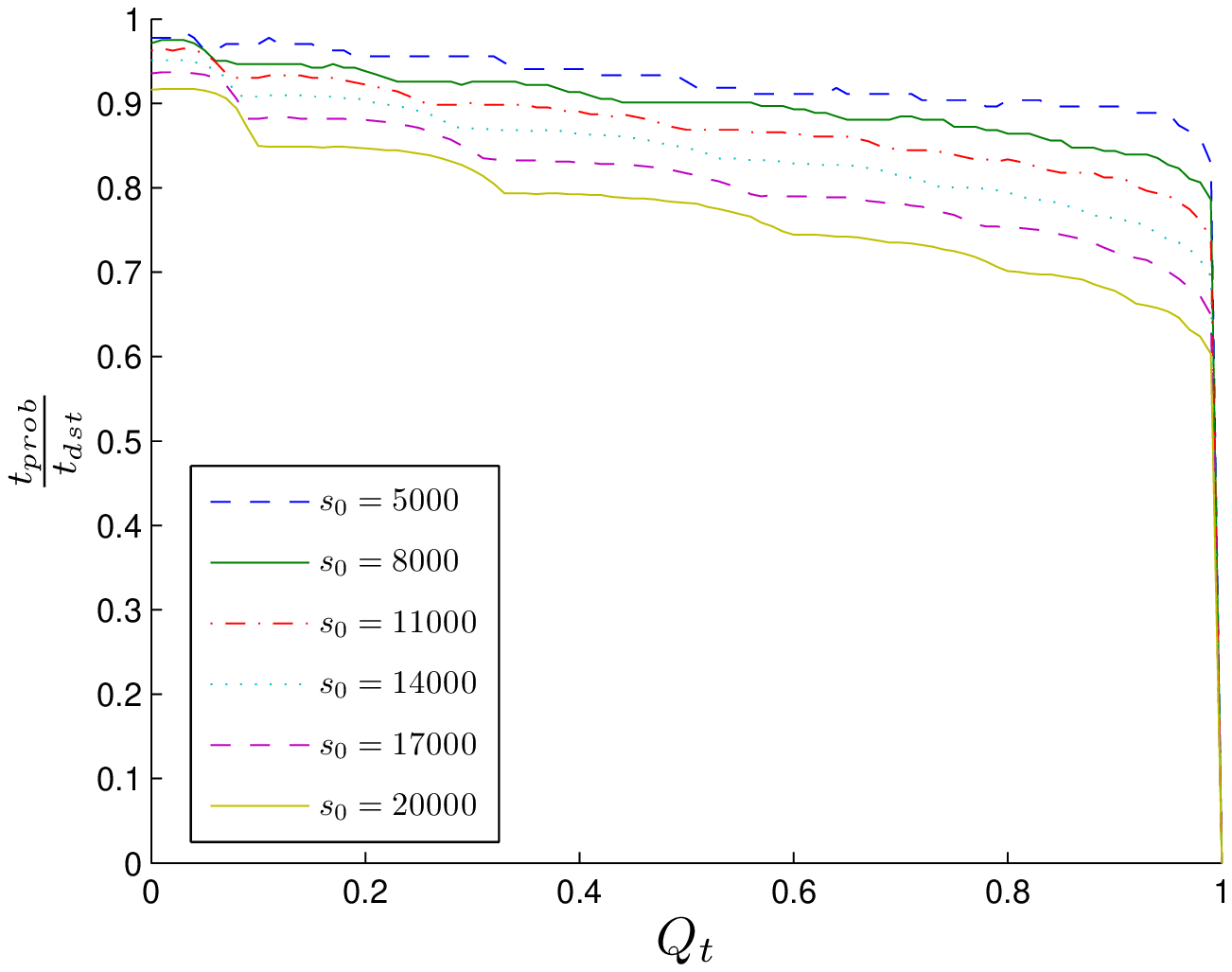}
}
\caption{
Figures~\subref{subfig.bounds.deltat2qt.s0} and~\subref{subfig.bounds.deltat2qt.k} compare the deterministic bound and the probabilistic bound as a function of the desired guaranteeing probability~-~$Q_t$ for various contaminated region's sizes and number of cleaning agents.}
\label{fig.DynamicProbLower.Use.boundsComparison2} 
\end{figure}

Notice that as illustrated in Figure~\ref{subfig.bounds.deltat2k}, as the number of agents increases the probabilistic and the deterministic bounds are more similar. This result is not surprising considering the method we presented. In our method, in order to make the bound tighter, we favor the lower values of $\delta$. As the number of agents increases the bound can be guaranteed in the desired probability with lower values of $\delta$ and as $\delta$ decreases the model becomes more similar to the deterministic one.

\subsection{Parameters Selection}
\label{p-selection}

One of the problems of the bound as brought in Section~\ref{d-bound} is the nature of the probabilistic bounds to decay to 0 (as shown in Figure~\ref{fig.DynamicProbLower.tp}), which caused due to the fact that the bound probability $Q_t$ is a product of each step's probability, $q_t$, and because as $t$ increase $Q_t$ decreases. One of the reasons which explains this problem is a bad selection of parameters~-~e.g. in the bound for the cleaning time (Eq.~\ref{eq.DynamicProbLower.LowerBound3} a selection of too small $\hat{S}$ will lead to fast decay in the probability. Furthermore, there exist trade-offs, when selecting the parameters' values, between the bound results and the the probability which guarantees its likelihood (e.g. see Figure~\ref{fig.DynamicProbLower.od}). 

One way to avoid this problem
 is by selecting $\hat{S}$ as big as possible, as in Figures~\ref{fig.DynamicProbLower.tp}and \ref{fig.DynamicProbLower.od}. As $\hat{S}$ increases the probability of each time step, $q_t$, increases and so the total probability $Q_t$. Another technique for eliminating this problem is by ``wrapping'' number of time steps into one, thus artificially decreasing the time and therefore decreasing the power of $q_t$ in $Q_t$ (in Lemma~\ref{lemma.DynamicProbLower.LowerBound}).

Another example of the trade-off in choosing the parameters can be shown in Figure~\ref{subfig.t-p} where $Q_t$ is illustrated for various values of the probability $p$. Interestingly, as $p$ decreases our confidence in the bound result is decreasing although we know that the the agents preforming the cleaning protocol have a better chance to successfully complete their work.

\begin{figure}
\centering
\subfigure[The bound probability $Q_t$ as a function of $\delta$ for various $p$ values] 
{
    \label{subfig.od-p}
    \includegraphics[height=8cm,width=9cm]{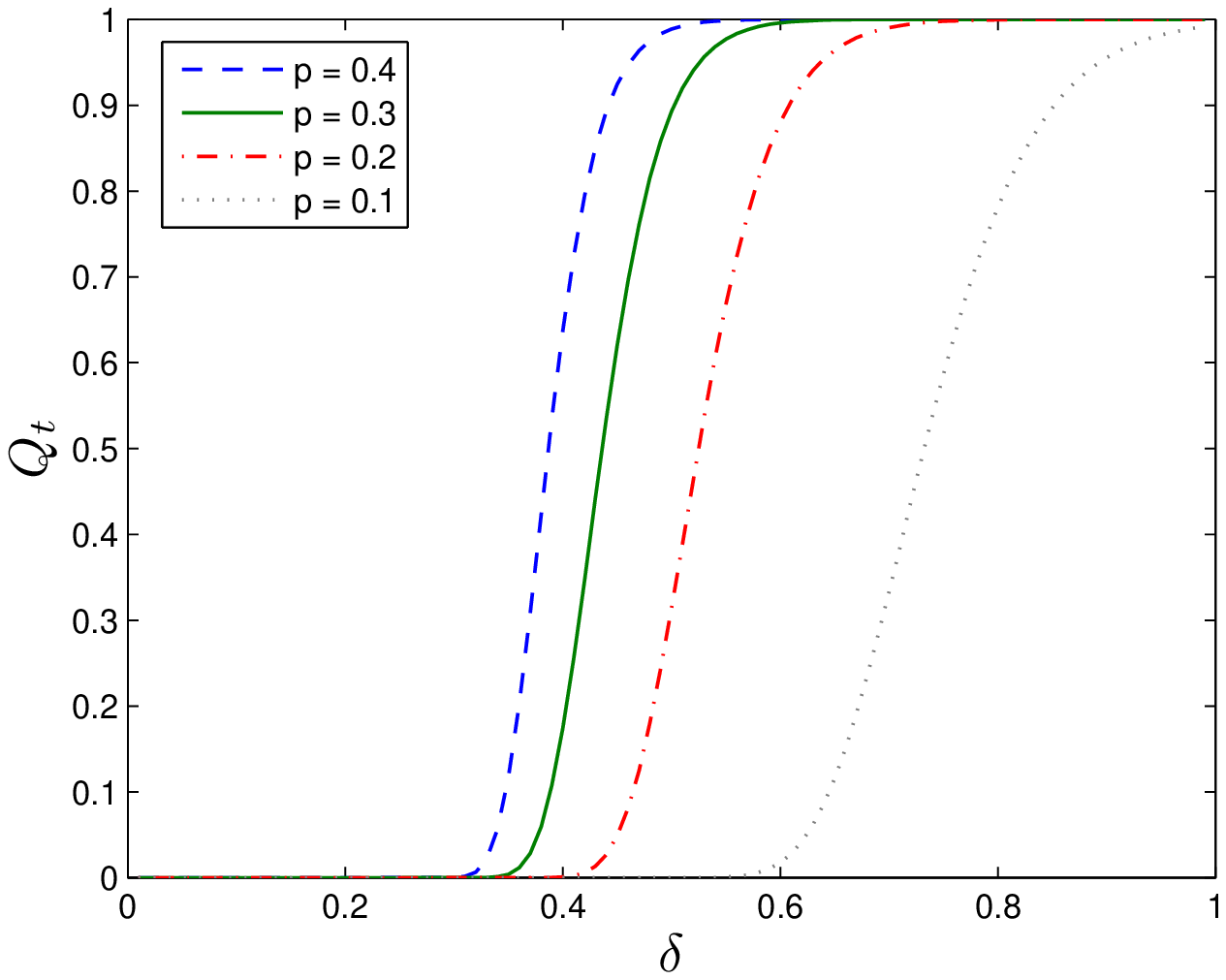}
}
\vspace{0.1cm}
\subfigure[The bound probability $Q_t$ as a function of $\delta$ for various numbers of agents] 
{
    \label{subfig.od-k}
    \includegraphics[height=8cm,width=9cm]{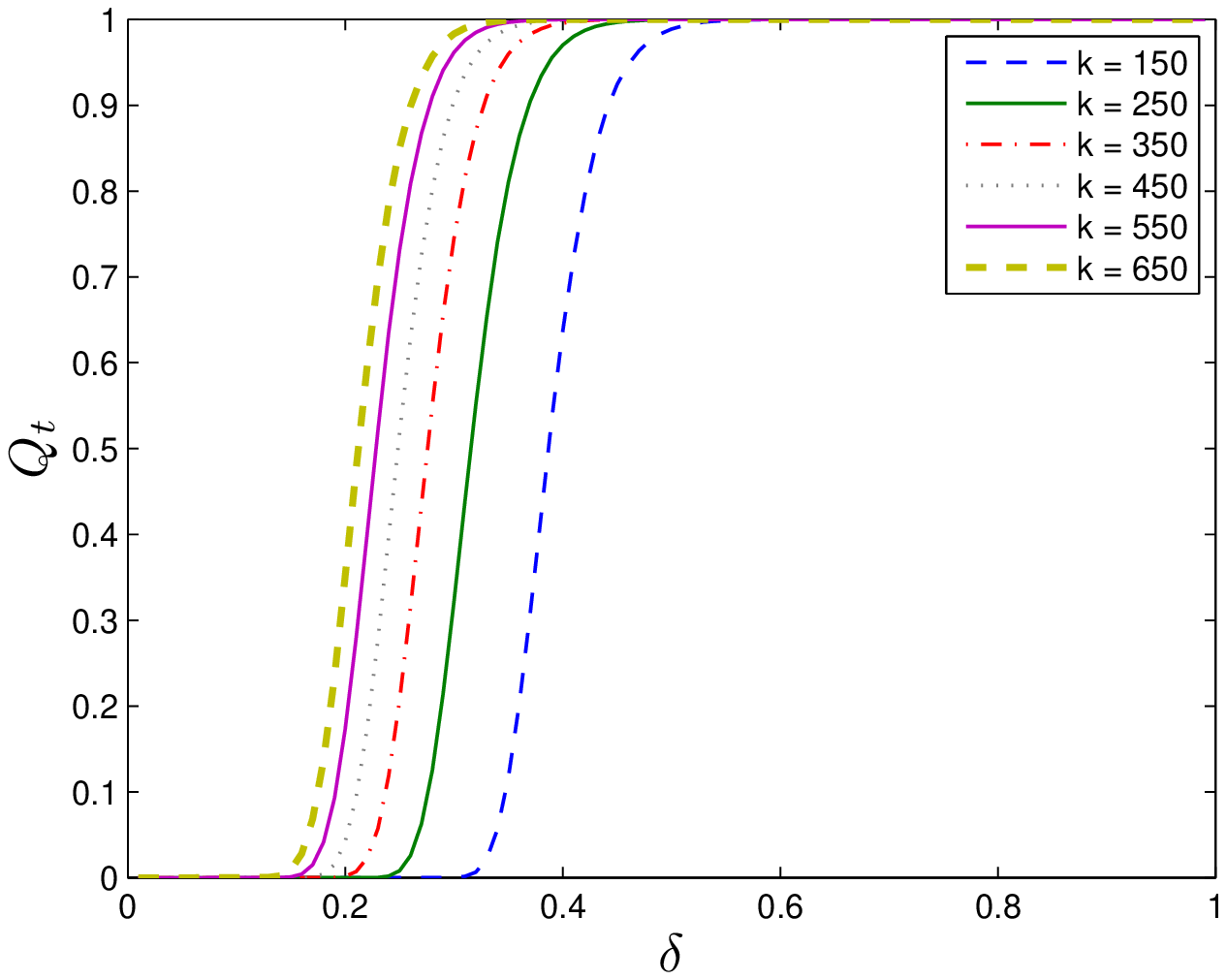}
}
\caption{Parameters selection. In \subref{subfig.od-p} we can see $Q_t$ for $S_0 = 20000$ and $k=150$ for the following of values of $p \in \left[0.4,0.3,0.2,0.1\right]$. In \subref{subfig.od-k} we can see $Q_t$ for $S_0 = 20000$ and $p = 0.4$ for the following of values of $k \in \left[150,250,350,450,550,650\right]$}
\label{fig.DynamicProbLower.od} 
\end{figure}

\section{Impossibility Result}
\label{i-result}
While the theoretical lower bound presented in Section~\ref{l-bound} can decrease the uncertainty whether a solution for the cleaning problem with certain number of agents exists, one might be interested in the opposite question, namely~---~how can we guarantee that a group of agents \emph{will not} be able to successfully accomplish the cleaning work (regardless of the cleaning protocol being used or the contaminated region's properties e.g. its shape and spreading probability). A first impossibility result for the deterministic case of the Cooperative Cleaners problem was published in~\cite{2010CSMAS}), where an initial size that is impossible to clean (using any protocol) was demonstrated. In this paper, we extend this result in order to be applicable for stochastically expending domains as well. Consequently, we will set the impossibility result with probabilistic restrictions as follows~:

\begin{theorem}
\label{theorem.DynamicProbLower.impossibility}
Using any cleaning protocol, $k$ agents cleaning a contaminated region 
, where each tile in the \emph{potential boundary} may be contaminate by already contaminated neighboring tiles with some probability $p$ in each time step, will not be able to cleat this contaminated region if~:
$$ S_0 > \left\lfloor \frac{k^2}{8 \cdot p^2}+k+\frac{1}{2} \right\rfloor $$
with the probability $Q_{t}$ for any time step $t$~-~i.e.~:
$$\forall{t} Pr\left[ S_{t} \geq S_0 \right] = Q_t$$
\end{theorem}

\begin{proof}
Firstly, we shell require that the contaminated region's size increases between each time step, guaranteeing us that the contaminated region's size will keep on growing, and thus impossible to be cleaned. Therefore we want that in each time step $t$ the size of the contaminated region will be bigger than the previous one~-~i.e.~: 
$$ S_{t+1}-S_{t}>0 $$
Using Theorem~\ref{theorem.DynamicProbLower.LowerBound} we know that~:
$$ S_{t+1} \geq S_{t} - k + \left\lfloor 2 \cdot \left( 1-\delta \right) p \cdot \sqrt{2 \cdot (S_{t} - k) - 1} \right\rfloor $$
and therefore we shell require that~:
$$ \left\lfloor 2 \cdot \left( 1-\delta \right) p \cdot \sqrt{2 \cdot (S_{t} - k) - 1} \right\rfloor -k >0 $$
Choosing, without loss of generality, $t=0$ and after some arithmetics, we see that~:
$$
S_0 > \left\lfloor \frac{k^2}{8\left(1-\delta \right)^2p^2}+k+\frac{1}{2} \right\rfloor 
$$
As $S_0$ as a function of $\delta$ is monotonically increasing and tends to infinity as $\delta$ tends to 1, we can lower bound $S_0$ with $\delta =0$, therefore~:
\begin{equation}
\label{eq.DynamicProbLower.impossible}
S_0 > \left\lfloor \frac{k^2}{8\cdot p^2}+k+\frac{1}{2} \right\rfloor 
\end{equation}
We would like this process to continue for $t$ time steps, thus applying the same method for all time steps and using Def.~\ref{def.DynamicProbLower.qt}, we get that~:
$$\forall{t} Pr\left[ S_{t} \geq S_0 \right] = Q_t$$.
\qed
\end{proof}

Notice that Theorem~\ref{theorem.DynamicProbLower.impossibility} produces two results~-~the first one is the minimal initial region's size $S_0$ which guarantees that the cleaning agents \emph{will not} be able to successfully accomplish the cleaning process and second one is the corresponding probability $Q_t$ in which this $S_0$ can be guaranteed. Also notice that in order to evaluate $Q_t$, one should use the appropriate $S_0$ and $\delta$.

\begin{figure}
\centering
\subfigure[The minimal initial region's size $S_0$ as a function of the number of agents $k$ for various probabilities.] 
{
    \label{subfig.imp.result.s0}
    \includegraphics[height=8cm,width=9cm]{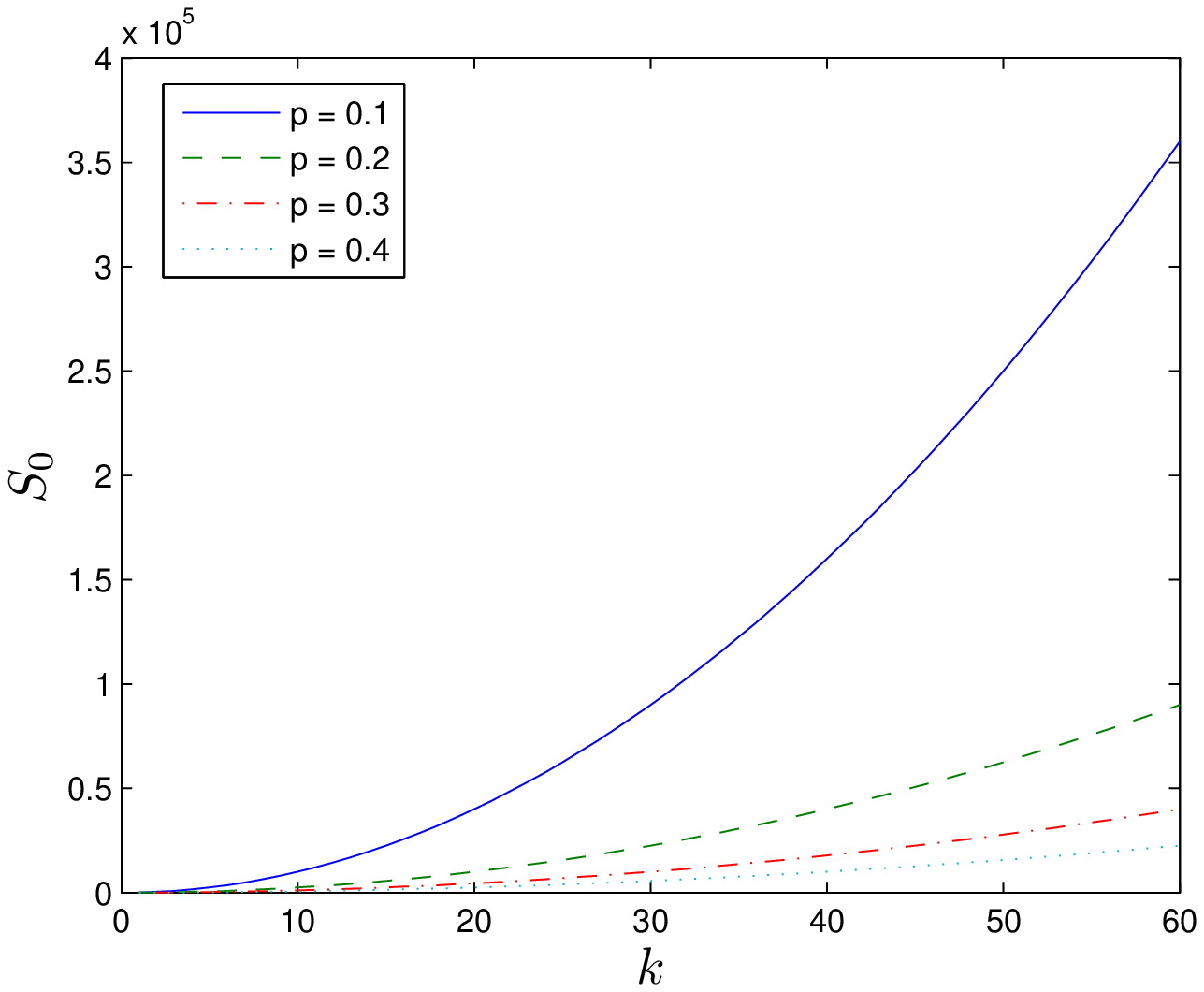}
}
\vspace{0.1cm}
\subfigure[The guaranteed probability $Q_t$ as a function of the number of agents $k$ for various probabilities.] 
{
    \label{subfig.imp.result.qt}
    \includegraphics[height=8cm,width=9cm]{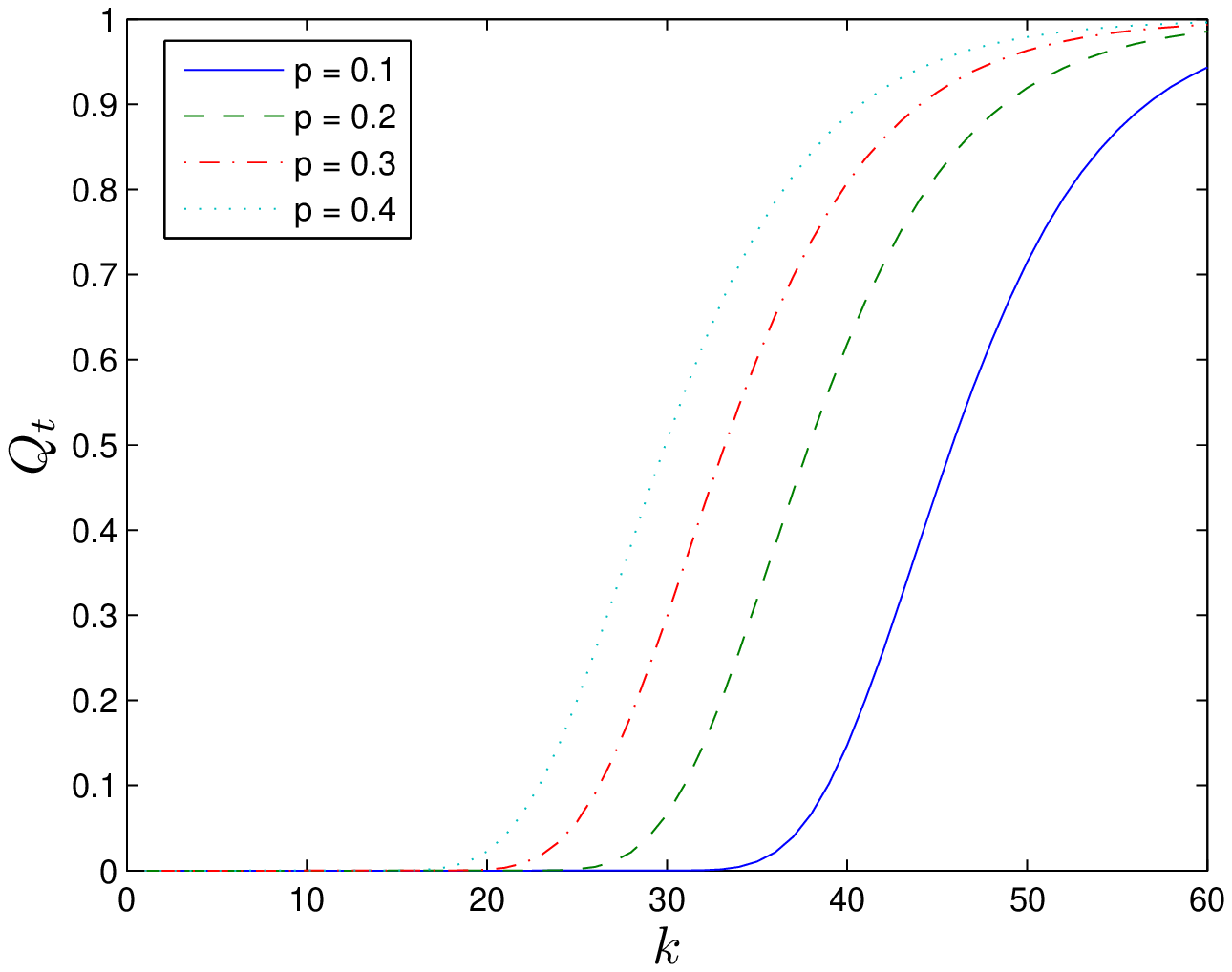}
}
\caption{An illustration of the impossibility result as presented in Theorem~\ref{theorem.DynamicProbLower.impossibility}. In \subref{subfig.imp.result.s0} we can see the minimal initial region's size $S_0$ for number of agents $k=[1..60]$ and spreading probability $p=[0.1,0.2,0.3,0.4]$. In \subref{subfig.imp.result.qt} we can see the corresponding probability $Q_t$.}
\label{fig.imp.results} 
\end{figure}

Interestingly, the results demonstrated in Figure.~\ref{fig.imp.results} of the impossibility result as presented in Theorem~\ref{theorem.DynamicProbLower.impossibility}, as the number of cleaning agents increases the probability which we can guarantee the minimal initial region's size $S_0$ also increases~-~in other words, although as the number of agents increases, the corresponding minimal initial region's size also increases and the probability in which we can guarantee that the agents will not be able to successfully clean the region increases as well.

\section{Experimental Results}
\label{e-results}

In previous work \cite{new-IJRR} a cleaning protocol for a group of K agents collaboratively cleaning an expending region was developed, called \sweep. The performance of the algorithm was analyzed in~\cite{new-IJRR}, both analytically and experimentally. We had implemented the algorithm for a revised environment - where stochastic changes take place, as defined throughout this paper. We have conducted an extensive observational analysis of the performance of this algorithm.
Exhaustive simulations were carried out, examining the cleaning activity of the protocol for various combinations of parameters~--~namely, number of agents, spreading probability (or spreading time in the deterministic model) and geometric features of the contaminated region. All the results were averaged over at least 1000 deferent runnings in order to get a statistical significance. In the deterministic model we average the results over the all the possible starting positions of the agents. Notice that, in order to minimize the running time, all running were stopped after some significant time and we consider these runnings as failure~--~i.e. these results are not included in the average calculation and not counted in the success percentage.

\begin{figure}
\centering
\subfigure[Average $T_{SUCCESS}$ as a function of the number of agents $k$] 
{
    \label{subfig.p0.02.t2k}
    \includegraphics[height=8cm,width=9cm]{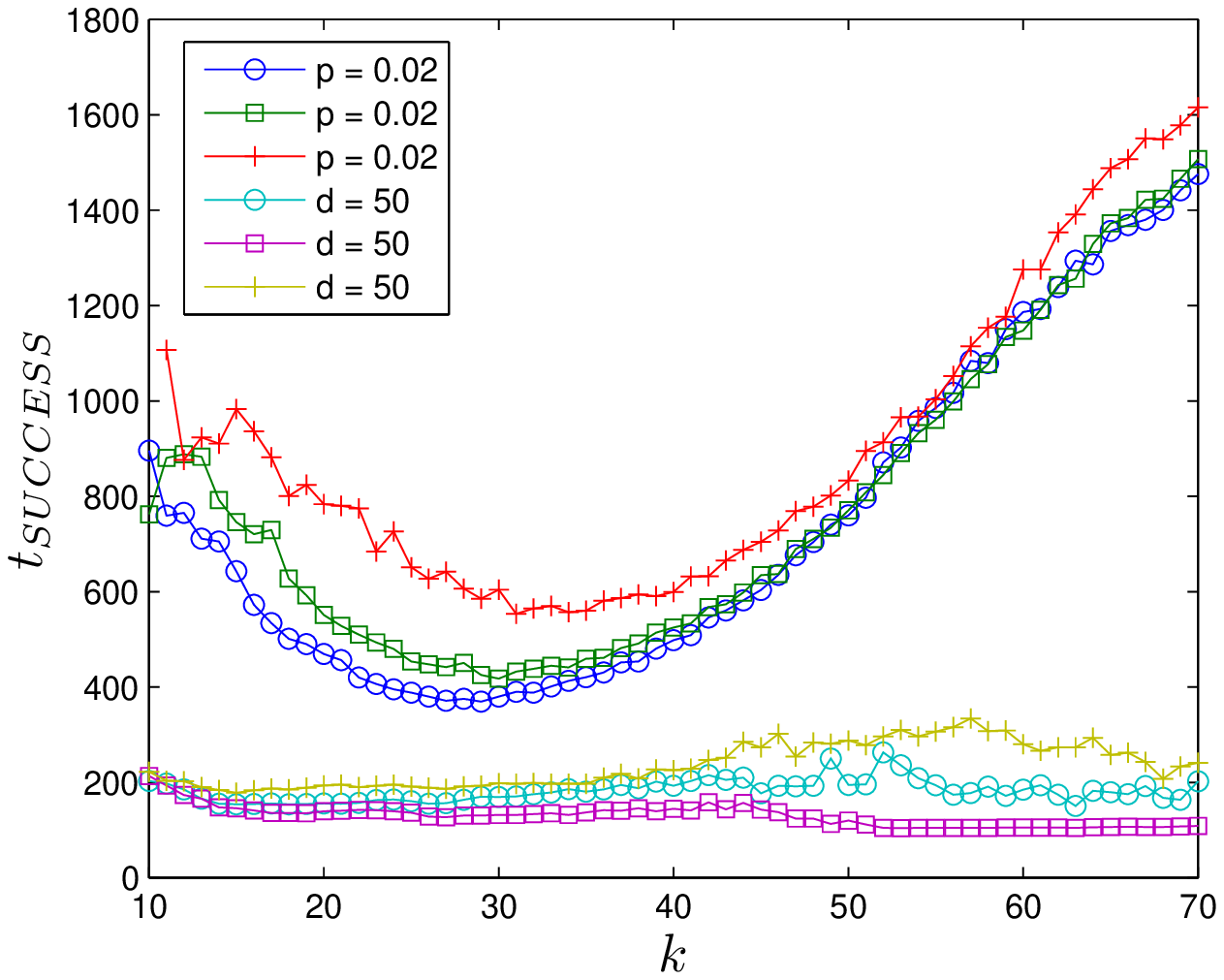}
}
\vspace{0.1cm}
\subfigure[Success percentage as a function of the number of agents $k$] 
{
    \label{subfig.p0.02.suc2k}
    \includegraphics[height=8cm,width=9cm]{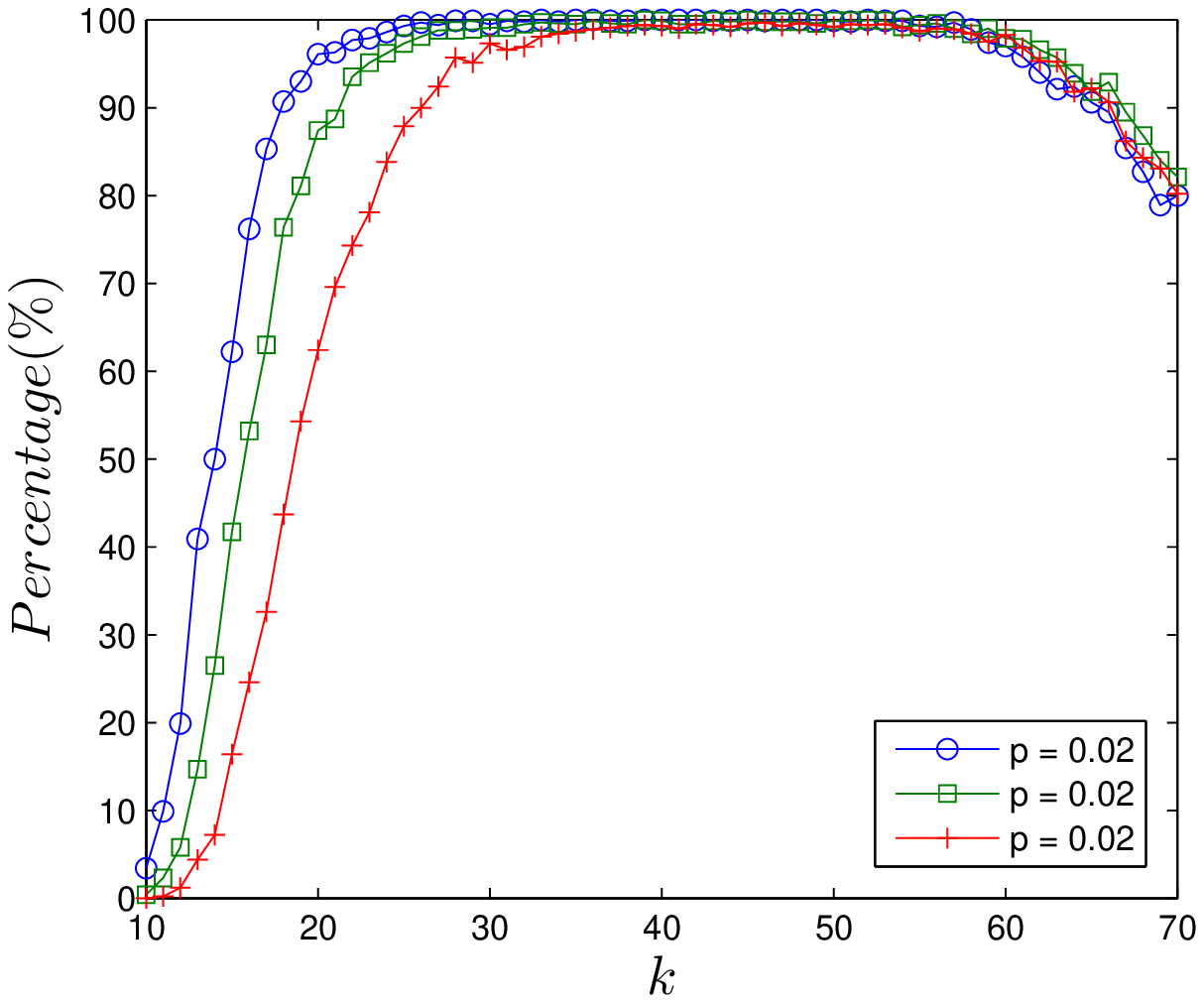}
}
\caption{Experimental results for various number of agents for spheric and squared contaminated region with starting size $S_0=500$ and with spreading probability of $p=0.02$ (notice that all the running were stopped after 3000 time steps). In \subref{subfig.p0.02.t2k} can see the results compared to the deterministic model results (with $d= \frac{1}{p}$ and 
in \subref{subfig.p0.02.suc2k} the success percentage}
\label{fig.results.p0.02} 
\end{figure}

Some of the experimental results are presented in Figure~\ref{fig.results.p0.02} comparing the probabilistic model and the deterministic model over three deferent shapes (circle, square and cross) with range of number of agents. Notice the interesting phenomenon, where adding more agents may cause to an increase in the cleaning time due to the delay caused by the agents synchronization in the \sweep protocol.

\section{Conclusions}
\label{conclu}
In this work we set the foundations of the stochastic model for the \emph{Cooperative Cleaning} problem and introduce, for the first time, the basic definitions describe this problem. We present two lower bounds on the contaminated region's size and on the cleaning time under the limitation of this probabilistic model and demonstrate an impossibility result on the number of agents which are essential for a successful completion cleaning a contaminated region.

One of the ways these results could be further enhanced would involve analyzing the transition process between several possible "states" of the system, as a Markov process. Once analyzing the process as a \emph{Markov's Chain}, we can get the stationary distribution of the process i.e. the probability to get to each one of the ending states (totally clean or impossible to clean).

It is also interesting to mention the similarity of this work to recent works done in the field of influence models in social networks. For example, in~\cite{PAN_AAAI2011} the authors demonstrate that a probabilistic local rule can efficiently simulate the spread of ideas in a social network. Combining this result with our work can generate a unique approach for analyzing dynamics of information flow in social networks.

\bibliography{Regev}

\end{document}